\let\accentvec\vec 
\documentclass[oribibl,envcountsame]{llncs}  
\usepackage{description_logic}
\usepackage{asymptote,stmaryrd,supertabular,colortbl}
\usepackage{pdflscape,enumerate}
\usepackage{enumitem,multirow}
\usepackage{multibib}
\newcites{appendix}{Appendix References}

\newif\ifasy
\asyfalse

\spnewtheorem{prop}[theorem]{Proposition}{\bfseries}{\itshape} 
\spnewtheorem{observation}[theorem]{Observation}{\bfseries}{\itshape}

\title{%
  \texorpdfstring{%
    Generalized Satisfiability\\ for the Description Logic \ALC%
  }{%
    Generalized Satisfiability for the Description Logic ALC%
  }%
}
\author{Arne Meier\protect\texorpdfstring{\inst{1}}{} and Thomas Schneider\protect\texorpdfstring{\inst{2}}{}}
\institute{%
  Leibniz Universit\"{a}t Hannover, Germany,~
  \email{meier@thi.uni-hannover.de}%
  \and
  University of Bremen, Germany,~ 
  \email{tschneider@informatik.uni-bremen.de}%
}
\authorrunning{A. Meier and T. Schneider}

\begin{document}

\maketitle
\thispagestyle{plain}

\begin{abstract}
  The standard reasoning problem, concept satisfiability, in the basic description logic \ALC
  is \PSPACE-complete, and it is \EXPTIME-complete in the presence of unrestricted axioms.
  Several fragments of \ALC, notably logics in the \FL, \EL, and DL-Lite families,
  have an easier satisfiability problem; sometimes it is even tractable.
  We classify the complexity of the standard satisfiability problems
  for all possible Boolean and quantifier fragments of \ALC in the presence of general axioms.
\end{abstract}

\section{Introduction}

Standard reasoning problems of description logics, such as satisfiability or subsumption,
have been studied extensively. Depending on the expressivity of the logic,
the complexity of reasoning for DLs between fragments of the basic DL \ALC and the OWL\,2 standard \SROIQ
is between trivial and \NEXPTIME.

For \ALC, concept satisfiability is \PSPACE-complete \cite{scsm91}.
In the presence of unrestricted axioms, it is \EXPTIME-complete due to the correspondence with propositional
dynamic logic \cite{pr78,vawo86,doma00}. Since the standard reasoning tasks are interreducible,
subsumption has the same complexity.

Several fragments of \ALC, such as logics in the \FL, \EL or DL-Lite families,
are well-understood. They usually restrict the use of Boolean operators and of quantifiers,
and it is known that their reasoning problems are often easier than for \ALC.
We now need to distinguish between satisfiability and subsumption because they are
no longer obviously interreducible if certain Boolean operators are missing.
Concept subsumption with respect to acyclic and cyclic terminologies, and even with general axioms,
is tractable in the logic \EL, which allows only conjunctions and existential restrictions,
\cite{baa03a,bra04}, and it remains tractable under a variety of extensions
such as nominals, concrete domains, role chain inclusions, and domain and range restrictions
\cite{bbl05,bbl08}. Satisfiability for \EL, in contrast, is trivial, i.e., every \EL-ontology is satisfiable.
However, the presence of universal quantifiers usually breaks tractability:
Subsumption in \FLzero, which allows only
conjunction and universal restrictions, is \coNP-complete \cite{neb90} and
increases to \PSPACE-complete with respect to cyclic terminologies \cite{baa96,kn03}
and to \EXPTIME-complete with general axioms \cite{bbl05,hof05}. In \cite{dlnhnm92,dlnn97},
concept satisfiability and subsumption for several logics below and above \ALC that extend \FLzero with
disjunction, negation and existential restrictions and other features,
is shown to be tractable, \NP-complete, \coNP-complete or \PSPACE-complete.
Subsumption in the presence of general axioms is \EXPTIME-complete
in logics containing both existential and universal restrictions plus conjunction or disjunction \cite{gimcwiko02},
as well as in \AL, where only conjunction, universal restrictions and unqualified existential restrictions
are allowed \cite{don03}.
In DL-Lite, where atomic negation, unqualified existential and universal restrictions, conjunctions
and inverse roles are allowed, satisfiability of ontologies is tractable \cite{cgllr05}.
Several extensions of DL-Lite are shown to have tractable and \NP-complete satisfiability
problems in \cite{ackz07,ackz09}.
The logics in the \EL and DL-Lite families are so important for (medical and database)
applications that OWL\,2 has two profiles that correspond to logics in these families.

This paper revisits restrictions to the Boolean operators in \ALC.
Instead of looking at one particular subset of $\{\sqcap,\sqcup,\neg\}$,
we are considering all possible sets of Boolean operators, and therefore our analysis includes
less commonly used operators such as the binary exclusive \emph{or} $\xor$.
Our aim is to find for \emph{every} possible combination of Boolean operators 
whether it makes satisfiability of the corresponding restriction of \ALC hard or easy.
Since each Boolean operator corresponds to a Boolean function---\ie, an $n$-ary function whose arguments and values are in $\{\ZERO,\ONE\}$---there are infinitely many sets of Boolean operators that determine fragments of \ALC.
The complexity of the corresponding concept satisfiability problems without theories
has already been classified in \cite{hescsc08} between being \PSPACE-complete, \coNP-complete,
tractable and trivial for all combinations of Boolean operators and quantifiers.

The tool used in \cite{hescsc08} for classifying the infinitely many satisfiability problems was
Post's lattice \cite{pos41},
which consists of all sets of Boolean functions closed under superposition. These sets directly correspond to
all sets of Boolean operators closed under composition. Similar classifications have been achieved
for satisfiability for classical propositional logic
\cite{le79}, Linear Temporal Logic \cite{bsssv07}, hybrid logic \cite{MMSTWW09}, and for constraint
satisfaction problems \cite{sch78,schnoor07}.

In this paper, we classify the concept satisfiability problems with respect to theories for \ALC fragments obtained by arbitrary sets of Boolean operators and quantifiers. We separate these problems into
\EXPTIME-complete, \NP-complete, \P-complete and \NL-complete, leaving only two single cases with non-matching upper and lower bound.
We will also put these results into the context of the above listed results for \ALC fragments.

This study extends our previous work in \cite{MS10} by matching upper and lower bounds
and considering restricted use of quantifiers.

\section{Preliminaries}
\label{sec:prelims}

\paragraph*{Description Logic.}

We use the standard syntax and semantics of \ALC \cite{DLHB}, with the Boolean operators $\dand$, $\sqcup$, $\neg$, $\top$, $\bot$
replaced by arbitrary operators $\fop{f}$ that correspond to Boolean functions $f : \{\ZERO,\ONE\}^n \to \{\ZERO,\ONE\}$
of arbitrary arity $n$. Let \CONC, \ROLE
and \IND be sets of atomic concepts, roles and individuals. Then the set of \emph{concept descriptions}, for short \emph{concepts}, is defined by
\[
    C := A \mid \fop{f}(C,\dots,C) \mid \exists R.C \mid \forall R.C,
\]
where $A \in \CONC$, $R \in \ROLE$, and $\fop{f}$ is a Boolean operator.
For a given set $B$ of Boolean operators, a \emph{$B$-concept} is a concept that uses only operators from $B$.
A \emph{general concept inclusion (GCI)} is an axiom of the form $C \sqsubseteq D$
where $C,D$ are concepts. We use ``$C \equiv D$'' as the usual syntactic sugar for ``$C \sqsubseteq D$ and $D \sqsubseteq C$''.
A \emph{TBox} is a finite set of GCIs without restrictions.
An \emph{ABox} is a finite set of axioms of the form $C(x)$ or $R(x,y)$,
where $C$ is a concept, $R \in \ROLE$ and $x,y \in \IND$.
An \emph{ontology} is the union of a TBox and an ABox. This simplified view suffices for our purposes.

An \emph{interpretation} is a pair $\calI = (\Delta^\calI, \cdot^\calI)$, where $\Delta^\calI$ is a nonempty set and $\cdot^\calI$ is a mapping from $\CONC$ to $\mathfrak{P}(\Delta^\calI)$, from $\ROLE$ to
$\mathfrak{P}(\Delta^\calI \times \Delta^\calI)$ and from $\IND$ to $\Delta^\calI$
that is extended to arbitrary concepts as follows:
\begin{align*}
  \fop{f}(C_1,\dots,C_n)^\calI &= \{x \in \Delta^\calI \mid f(\|x \in C_1^\calI\|, \dots, \|x \in C_n^\calI\|) = \ONE\}, \\
                               & \quad~\,\text{where~} \|x \in C_1^\calI\|=\ONE \text{~if~} x \in C_1^\calI
                                       \text{~and~} \|x \in C_1^\calI\|=\ZERO \text{~if~} x \notin C_1^\calI,    \\
  \exists R.C^\calI            &= \{x\in\Delta^\calI \mid \{y\in C^\calI \mid (x,y)\in R^\calI\} \neq \emptyset\},\\
  \forall R.C^\calI            &= \{x\in\Delta^\calI \mid \{y\in C^\calI \mid (x,y)\notin R^\calI\} = \emptyset\}.
\end{align*}
An interpretation \calI \emph{satisfies} the axiom $C \sqsubseteq D$, written $\calI \models C \sqsubseteq D$, if
$C^\calI \subseteq D^\calI$. Furthermore, $\calI$ satisfies $C(x)$ or $R(x,y)$ if $x^\calI \in C^\calI$
or $(x^\calI,y^\calI) \in R^\calI$. An interpretation \calI satisfies a TBox (ABox, ontology)
if it satisfies every axiom therein. It is then called a \emph{model} of this set of axioms.

Let $B$ be a finite set of Boolean operators and $\calQ \subseteq \{\exists,\forall\}$. 
We use $\Conc_\calQ(B)$, $\TBOX_\calQ(B)$ and $\ONT_\calQ(B)$ to denote the set of all concepts, TBoxes
and ontologies that use operators in $B$ only and quantifiers from \calQ only.
The following decision problems are of interest for this paper.
\begin{description}
  \itemsep4pt
  \item[\textbf{Concept satisfiability $\ALCCSAT_\calQ(B)$}:]~\\
    Given a concept $C \in \Conc_\calQ(B)$, is there an interpretation \calI \st $C^\calI\neq\emptyset$\,?
  \item[\textbf{TBox satisfiability $\ALCTSAT_\calQ(B)$}:]~\\
    Given a TBox $\calT \subseteq \TBOX_\calQ(B)$, is there an interpretation \calI \st $\calI \models \calT$\,?
  \item[\textbf{TBox-concept satisfiability $\ALCTCSAT_\calQ(B)$}:]~\\
    Given $\calT \subseteq \TBOX_\calQ(B)$ and $C \in \Conc_\calQ(B)$, is there an \calI \st $\calI \models \calT$
    and $C^\calI\neq\emptyset$\,?
  \item[\textbf{Ontology satisfiability $\ALCOSAT_\calQ(B)$}:]~\\
    Given an ontology $\calO \subseteq \ONT_\calQ(B)$, is there an interpretation \calI \st $\calI \models \calO$\,?
  \item[\textbf{Ontology-concept satisfiability $\ALCOCSAT_\calQ(B)$}:]~\\
    Given $\calO \subseteq \ONT_\calQ(B)$ and $C \in \Conc_\calQ(B)$, is there an \calI \st $\calI \models \calO$
    and $C^\calI\neq\emptyset$\,?
\end{description}
By abuse of notation, we will omit set parentheses and commas when stating \calQ explicitly, as in $\ALCTSAT_\exall(B)$.
The above listed decision problems are interreducible independently of $B$ and \calQ in the following way:
\begin{align*}
  & \ALCCSAT_\calQ(B) \leqlogm \ALCOSAT_\calQ(B) \\
  & \ALCTSAT_\calQ(B) \leqlogm \ALCTCSAT_\calQ(B) \leqlogm \ALCOSAT_\calQ(B) \equivlogm \ALCOCSAT_\calQ(B)
\end{align*}
A concept $C$ is satisfiable iff the ontology $\{C(a)\}$ is satisfiable,
for some individual $a$;~ a terminology $\calT$ is satisfiable iff a fresh atomic concept $A$ is satisfiable w.r.t. $\calT$; $C$ is satisfiable w.r.t.\ $\calT$ iff $\calT \cup \{C(a)\}$ is satisfiable, for a fresh individual $a$.

Some reductions in the main part of the paper consider another decision problem which is called \emph{subsumption} (\SUBS)
and is defined as follows:
Given a TBox \calT and two atomic concepts $A,B$, does every model of \calT satisfy $A \sqsubseteq B$\,?

\paragraph*{Complexity Theory.}

We assume familiarity with the standard notions of complexity theory as, \eg, defined in \cite{pap94}.
In particular, we will make use of the classes $\NL$, $\P$, $\NP$, $\coNP$, and $\EXPTIME$,
as well as logspace reductions $\leqlogm$.

\paragraph*{Boolean operators.}
  This study is complete with respect to Boolean operators, which correspond to Boolean functions.
  The table below lists all Boolean functions that we will mention, together with the associated
  DL operator where applicable.
  \begin{figure}
  	\centering
    \begin{small}
      \begin{tabular}{@{~~~}l@{~~~~~}l@{~~~~~}l@{~~~}}
        \hline\rule{0pt}{8pt}%
        Function symbol & Description                                    & DL operator symbol \\
        \hline\rule{0pt}{8pt}%
        \cZero, \cOne   & constant  \ZERO, \ONE                          & $\bot$, $\top$     \\
        \AND, \OR       & binary conjunction/disjunction $\land$, $\lor$ & $\sqcap$, $\sqcup$        \\
        \Neg            & unary negation $\bar{\cdot}$                   & $\neg$          \\
        \Xor            & binary exclusive \emph{or} $\xor$              & $\dxor$          \\
        \Andor          & $x \land (y \lor z)$                           &                 \\
        \SD             & $(x \land \overline{y}) \lor (x \land \overline{z}) \lor (\overline{y} \land \overline{z})$ & \\
        \Equiv          & binary equivalence function                    &                 \\
        \hline
      \end{tabular}
    \end{small}
    \caption{Boolean functions with description and corresponding DL operator symbol.}
  \end{figure}

  A set of Boolean functions is called a \emph{clone} if it contains all projections
  (also known as identity functions, the eponym of the $\CloneI$-clones below)
  and is closed under composition (also referred to as superposition). The lattice of all clones has been established in \cite{pos41},
  see \cite{bocrrevo03} for a more succinct but complete presentation. Via the inclusion structure,
  lower and upper complexity bounds can be carried over to higher and lower clones under certain
  conditions. We will therefore state our results for minimal and maximal clones only,
  together with those conditions.

  Given a finite set $B$ of functions, the smallest clone containing $B$ is denoted by $[B]$.
  The set $B$ is called a \emph{base} of $[B]$, but $[B]$ often has other bases as well.
  For example, nesting of binary conjunction yields conjunctions
  of arbitrary arity. The table below lists all clones that we will refer to, using
  the following definitions.
  A Boolean function $f$ is called \emph{self-dual} if
  $f(\overline{x_1},\dots,\overline{x_n}) = \overline{f(x_1,\dots,x_n)}$,
  \emph{$c$-reproducing} if $f(c,\dots,c) = c$ for $c \in \{\ZERO,\ONE\}$,
  and \emph{$c$-separating} if there is an $1\leq i\leq n$ \st for each $(b_1,\dots,b_n)\in f^{-1}(c)$,
  it holds that $b_i=c$.

  \begin{figure}
  \centering
    \begin{small}
      \begin{tabular}{@{~~~}l@{~~~~~}l@{~~~~~}l@{~~~}}
        \hline\rule{0pt}{8pt}%
        Clone                    & Description                             & Base                               \\
        \hline\rule{0pt}{8pt}%
        $\CloneBF$               & all Boolean functions                   & $\{\AND, \Neg\}$                   \\
        $\CloneR_0$, $\CloneR_1$ & $\ZERO$-, $\ONE$-reproducing functions  & $\{\AND,\Xor\}$,~ $\{\OR,\Equiv\}$ \\
        $\CloneM$                & all monotone functions                  & $\{\AND, \OR, \cZero, \cOne\}$ \\
        $\CloneS_1$              & \ONE-separating functions               & $\{x \wedge \overline{y}\}$ \\
        $\CloneS_{11}$           & \ONE-separating, monotone functions     & $\{\Andor,\cZero\}$ \\
        $\CloneD$                & self-dual functions                     & $\{\SD\}$ \\
        $\CloneL$                & affine functions                        & $\{\Xor,\cOne\}$ \\
        $\CloneL_0$              & affine, \ZERO-reproducing functions     & $\{\Xor\}$ \\
        $\CloneL_3$              & affine, \ZERO- and \ONE-reproducing functions & $\{x\;\Xor\; y\;\Xor\; z\;\Xor\;\cOne\}$ \\
        $\CloneE_0$, $\CloneE$   & conjunctions and $\cZero$ (and $\cOne$) & $\{\AND,\cZero\}$,~ $\{\AND,\cZero,\cOne\}$   \\
        $\CloneV_0$, $\CloneV$   & disjunctions and $\cZero$ (and $\cOne$) & $\{\OR,\cZero\}$,~ $\{\OR,\cZero,\cOne\}$   \\
        $\CloneN_2$, $\CloneN$   & negation (and $\cOne$)                  & $\{\Neg\}$,~ $\{\Neg,\cOne\}$ \\
        $\CloneI_0$, $\CloneI$   & $\cZero$ (and $\cOne$)                  & $\{\cZero\}$,~ $\{\cZero,\cOne\}$ \\
        \hline
      \end{tabular}
    \end{small}
    \caption{List of all relevant clones in this paper with their standard bases.}
  \end{figure}

  From now on, we will use $B$ to denote a finite set of Boolean operators.
  Hence, $[B]$ consists of all operators obtained by nesting operators
  from $B$. By abuse of notation, we will denote operator sets
  with the above clone names when this is not ambiguous.
  Furthermore, we call a Boolean operator corresponding to a monotone
  (self-dual, $\ZERO$-reproducing, $\ONE$-reproducing, $\ONE$-separating) function
  a monotone (self-dual, $\bot$-reproducing, $\top$-reproducing, $\top$-separating) operator.

  The following lemma will help restrict the length
  of concepts in some of our reductions. It shows that for
  certain operator sets $B$, there are always short concepts representing the
  operators $\dand$, $\sqcup$, or $\lnot$, respectively.  Points (2) and (3)
  follow directly from the proofs in \cite{le79}, Point (1) is
  Lemma~1.4.5 from \cite{schnoor07}.

\begin{lemma}\label[lemma]{lem:lewis-schnoor}
  Let $B$ be a finite set of Boolean operators.
  \begin{enumerate}
  \item
    If $\CloneV\subseteq[B]\subseteq\CloneM$
    ($\CloneE\subseteq[B]\subseteq\CloneM$, resp.), then there exists a $B$-concept $C$ such that $C$ is equivalent to
    $A_1\sqcup A_2$ ($A_1\dand A_2$, resp.) and each of the atomic concepts $A_1,A_2$ occurs exactly once in $C$.
  \item
    If $[B]=\CloneBF$, then there are $B$-concepts $C$ and
    $D$ such that $C$ is equivalent to $A_1\sqcup A_2$, $D$ is equivalent to $A_1\dand A_2$, 
    and each of the atomic concepts $A_1,A_2$ occurs in $C$ and $D$ exactly once.
  \item
    If $\CloneN\subseteq[B]$, then there is a $B$-concept
    $C$ such that $C$ is equivalent to $\neg A$ and the atomic concept $A$ occurs in $C$ only once.
  \end{enumerate}
\end{lemma}

%
%
%
%

\paragraph*{Auxiliary results.}
The following lemmata contain technical results that will be useful to formulate our main results.
We use $\STARSAT_\calQ(B)$ to speak about any of the four satisfiability problems $\ALCTSAT_\calQ(B),\ALCTCSAT_\calQ(B),\ALCOSAT_\calQ(B)$ and $\ALCOCSAT_\calQ(B)$ introduced above; for the three problems having the power to speak about a single individual, we abuse this notion and write $\cSTARSAT_\calQ(B)$ for the problems $\STARSAT_\calQ(B)$ without $\TSAT_\calQ(B)$.

\begin{lemma}[\cite{MS10}]\label[lemma]{lem:topbot-always-above-neg}
  Let $B$ be a finite set of Boolean operators s.t. $\CloneN_2\subseteq[B]$ and $\calQ \subseteq \{\exists,\forall\}$.
 Then it holds that $\STARSAT_\calQ(B)\equivlogm\STARSAT_\calQ(B\cup\{\true,\false\})$.
\end{lemma}
\begin{Proof}
  It is easy to observe that the concepts $\true$ and $\false$ can be simulated by fresh atomic concepts
  $T$ and $B$, using the axioms $\lnot T\dsub T$ and $B\dsub\lnot B$. 
\end{Proof}

\begin{lemma}[\cite{MS10}]\label[lemma]{lem:TCSAT_reduces_to_TSAT_with_true}
  Let $B$ be a finite set of Boolean operators and $\calQ \subseteq \{\exists,\forall\}$.
  Then it holds that $\ALCTCSAT_\calQ(B)\leqlogm\ALCTSAT_{\calQ\cup\{\exists\}}(B\cup\{\true\})$.
\end{lemma}
\begin{Proof}
It can be easily shown that $(C,\calT)\in\ALCTCSAT_\calQ(B)$ iff $(\calT\cup\{\true\dsub\exists R.C\})\in\ALCTSAT_\calQ(B\cup\{\true\})$, where $R$ is a fresh role. For ``$\Rightarrow$'' observe that for the satisfying interpretation $\calI=(\Delta^\calI,\cdot^\calI)$ there must be an individual $w'$ where $C$ holds and then from every individual $w\in\Delta^\calI$ there can be an $R$-edge from $w$ to $w'$ to satisfy $\calT\cup\{\true\dsub\exists R.C\}$. For ``$\Leftarrow$'' note that for a satisfying interpretation $\calI=(\Delta^\calI,\cdot^\calI)$ all axioms in $\calT\cup\{\true\dsub\exists R.C\}$ are satisfied. In particular the axiom $\true\dsub\exists R.C$. Hence there must be at least one individual $w'$ \st $w'\models C$. Thus $\calI\models\calT$ and $C^\calI\supseteq\{w'\}\neq\emptyset$.
\end{Proof}

\medskip\noindent
Furthermore, we observe that, for each set $B$ of Boolean operators with $\true,\false\in[B]$, we can simulate the negation of an atomic concept using a fresh atomic concept $A$ and role $R_A$: if we add the axioms $A\equiv\exists R_A.\true$ and $A'\equiv\forall R_A.\false$ to the given terminology $\calT$, then each model of \calT has to interpret $A'$ as the complement of $A$.

In order to generalize complexity results from $\STARSAT_\calQ(B_1)$ to $\STARSAT_\calQ(B_2)$ for \emph{arbitrary} bases $B_2$ of $[B_1]$,
we need the following lemma.
\begin{lemma}[\cite{MS10}]\label[lemma]{lem:Base_Independence}
Let $B_1,B_2$ be two sets of Boolean operators \st $[B_1]=[B_2]$, and let $\calQ \subseteq \{\exists,\forall\}$. 
Then $\STARSAT_\calQ\leqlogm\STARSAT_\calQ(B_2)$.
\end{lemma}
\begin{Proof}
According to \cite[Theorem 3.6]{hescsc08}, we translate for any given instance each concept (hence each side of an axiom) into a Boolean circuit over the basis $B_1$. This circuit can be easily transformed into a circuit over the basis $B_2$. This new circuit will be expressed by several new axioms that are constructed in the style of the formulae in \cite{hescsc08}:
\begin{itemize}
        \item For input gates $g$, we add the axiom $g\equiv x_i$.
        \item If $g$ is a gate computing the Boolean operator $\circ$ and $h_1,\dots,h_n$ are the respective predecessor gates in this circuit, we add the axiom $g\equiv\circ(h_1,\dots,h_n)$.
        \item For $\exists R$-gates $g$, we add the axiom $g\equiv \exists R.h$.
        \item Analogously for $\forall R$-gates.
\end{itemize}

For each axiom $A\dsub B$, let $g_{out}^A$ and $g_{out}^B$ be the output gates of the appropriate circuits. Then we need to add one new axiom $g_{out}^A\dsub g_{out}^B$ to ensure the axiomatic property of $A\dsub B$. For a concept $C$ in the input (relevant for the problems $\ALCTCSAT_\calQ,\ALCOCSAT_\calQ$), its translation is mapped to the respective out-gate $g_{out}^C$.

This reduction is computable in logarithmic space and its correctness can be shown in the same way as in the Proof of \cite[Theorem 3.6]{hescsc08}.
\end{Proof}


\noindent The idea for the following lemma goes back to Lewis \cite{le79}.

\begin{lemma}[Lewis Trick]\label[lemma]{lem:Lewis_Trick2}
Let $B$ be a set of Boolean operators and $\calQ\subseteq\{\forall,\exists\}$. Then it holds that $\TSAT_\calQ(B\cup\{\true\})\leqlogm\TCSAT_\calQ(B)$.
\end{lemma}
\begin{Proof}
Let $\SC(\calT)$ be the set of all (sub-)concepts occurring in \calT.
For every $C \in \SC(\calT)$, we use $C_T$ to denote $C$ with all occurrences of $\top$ replaced by $T$.
Furthermore, we write $\calT_T$ for $\{C_T \sqsubseteq D_T \mid C \sqsubseteq D \in \calT\}$.

We claim that $\calT\in\TSAT_\calQ(B) \iff (\calT',T)\in\TCSAT_\calQ(B)$, where
\begin{align*}
\calT' = \calT_T\cup\set{C_T \sqsubseteq T}{C \in \SC(\calT)}.
\end{align*}

For the direction "$\Rightarrow$" observe that for any interpretation $\calI=(\Delta^\calI,\cdot^\calI)$ with $\calI\models\calT$, we can set $T^\calI=\Delta^\calI$ and then have $\calI\models\calT'$ and obviously $\calT^\calI\neq\emptyset$.

Now consider the opposite direction "$\Leftarrow$". Let $\calI=(\Delta^\calI,\cdot^\calI)$ be an interpretation s.t. $\calI\models\calT'$ and $T^\calI\neq\emptyset$.
We construct $\calJ$ from $\calI$ via restriction to $T^\calI$, i.e., $\Delta^\calJ = T^\calI$,
$A^\calJ = A^\calI \cap T^\calI$ for atomic concepts $A$, and $R^\calJ = R^\calI \cap (T^\calI \times T^\calI)$
for roles $R$. We claim the following:

\smallskip\noindent\emph{Claim.}
For every individual $x \in T^\calI$ and every (sub-)concept $C$ occurring in \calT, it holds that
$x \in C_T^\calI$ if and only if $x \in C^\calJ$.

\smallskip\noindent
This claim implies that $\calJ \models \calT$: for any $x \in \Delta^\calJ = T^\calI$
and any axiom $D \sqsubseteq E \in \calT$, we have that $x \in D^\calJ$ implies $x \in D_T^\calI$
due to the claim, which implies $x \in E_T^\calI$ because $\calI \models \calT'$,
which implies $x \in E^\calJ$ due to the claim.

\smallskip\noindent\emph{Proof of Claim.}
We proceed by induction on the structure of $C$. The base case includes atomic $C$ as well as $\top$ and $\bot$,
and follows from the construction of \calJ.

For the induction step, we consider the following cases.
\begin{itemize}
  \item
    In case $C = \fop{f}(C^1,\dots,C^n)$, where $\fop{f}$ is an arbitrary $n$-ary boolean operator
    corresponding to an $n$-ary Boolean function $f$, and the $C^i$ are smaller subconcepts of $C$, the following holds.
    \begin{xalignat*}{2}
      x \in C_T^\calI
      & ~~~\text{iff}~~~ f(\|x \in (C^1_T)^\calI\|, \dots, \|x \in (C^n_T)^\calI\|) = \ONE & & \text{def.\ of satisfaction} \\
      & ~~~\text{iff}~~~ f(\|x \in (C^1)^\calJ\|, \dots, \|x \in (C^n)^\calJ\|) = \ONE     & & \text{induction hypothesis}  \\
      & ~~~\text{iff}~~~ x \in C^\calJ                                                     & & \text{def.\ of satisfaction}
    \end{xalignat*}
  \item
    In case $C = \exists R.D$, the following holds.
    \begin{align*}
      x \in C_T^\calI
      & ~~~\text{iff}~~~ \text{for some~} y \in \Delta^\calI: (x,y) \in R^\calI \text{~and~} y \in D_T^\calI \\
      & ~~~\text{iff}~~~ \text{for some~} y \in T^\calI: (x,y) \in R^\calI \text{~and~} y \in D_T^\calI \\
      & ~~~\text{iff}~~~ \text{for some~} y \in T^\calI: (x,y) \in R^\calJ \text{~and~} y \in D^\calJ \\
      & ~~~\text{iff}~~~ x \in C^\calJ
    \end{align*}
    The first equivalence is due to the definition of satisfaction. The second's ``$\Rightarrow$'' direction is due to the additional axiom $D_T \sqsubseteq T$ in $\calT'$, while the ``$\Leftarrow$'' direction is obvious. The third
    equivalence is again due to the definition of satisfaction and the construction $\Delta^\calJ = T^\calI$.
  \item
    In case $C = \forall R.D$, we rewrite to $C = \neg\exists R.\neg D$, apply the previous two cases, and rewrite back.
\end{itemize}
\end{Proof}

\begin{lemma}[Contraposition]\label[lemma]{lem:contraposition}
Let $B$ be a set of Boolean functions and $\calQ\subseteq\{\exists,\forall\}$. Then
\begin{enumerate}
  \item $\TSAT_\calQ(B)\leqlogm\TSAT_{\dual{\calQ}}(\dual{B})$, and 
  \item $\TCSAT_\calQ(B)\leqlogm\TCSAT_{\dual{\calQ}}(\dual{B}\cup\{\false,\dand\})$,
\end{enumerate}
where $\dual{B}:=\set{\dual{f}}{f\in B}$ and $\dual{\calQ}=\{\dual{q}\mid q\in\calQ\}$ for $\dual{\exists}:=\forall$ and $\dual{\forall}=\exists$.
\end{lemma}
\begin{Proof} Let $B$ be a set of Boolean functions and $\calQ\subseteq\{\exists,\forall\}$. Let $\calT\in\TBOX_\calQ(B)$ be a terminology. 
\begin{enumerate}
  \item Then it holds that $\calT\in\TSAT_\calQ(B)$ iff $\calT^{\text{con}}\in\TSAT_{\dual\calQ}(\dual B)$, where 
  $$
  \calT^{\text{con}}:=\set{D^\lnot\dsub C^\lnot}{(C\dsub D)\in\calT},
  $$
  and $C^\lnot$ is $C$ in negation normalform (all negations are moved inside s.t. they are in front of atomic concepts) and the negated atomic concepts $\lnot A$ are replaced with fresh atomic concepts $A'$. Because of the negation normalform all functions are mapped to their dual and the quantifiers are expressed via their dual one. Therefore note that $C\dsub D \iff \lnot D\dsub \lnot C$.
  \item Here we need the operators $\false$ and $\dand$ to ensure that the input concept $C$ is not instantiated by the same individual as $C'$. Now observe that it holds that
  $(C,\calT)\in\TCSAT_\calQ(B)$ iff $(C,\calT^{\text{con}}\cup\{C\dand C'\dsub\false\})\in\TCSAT_{\dual\calQ}(\dual B)$, where $\calT^{\text{con}}$ is as in (1.).
\end{enumerate}
\end{Proof}

\paragraph*{Known complexity results for \CSAT.}
In \cite{hescsc08},
the complexity of concept satisfiability has been classified for modal logics
corresponding to all fragments of \ALC with arbitrary combinations of
Boolean operators and quantifiers: $\ALCCSAT_{\calQ}(B)$
with $\calQ \subseteq \{\exists,\forall\}$ is either \PSPACE-complete, \co\NP-complete,
or in \P. Some of the latter cases are trivial, i.e., every concept in such a fragment
is satisfiable. These results generalize known complexity results for \ALE
and the \EL and \FL families. On the other hand, results for \ALU and the
DL-Lite family cannot be put into this context because they only allow
unqualified existential restrictions. See \cite{MS10} for a more detailed discussion.

\section{Complexity Results for \TSAT, \TCSAT, \OSAT, \OCSAT}
In this section we will almost completely 
classify the above mentioned satisfiability problems for their tractability with respect to sub-Boolean fragments and put them into context with existing results for fragments of \ALC.

We use $\STARSAT_\calQ(B)$ to speak about any of the four satisfiability problems $\ALCTSAT_\calQ(B),\ALCTCSAT_\calQ(B),\ALCOSAT_\calQ(B)$ and $\ALCOCSAT_\calQ(B)$ introduced above; for the three problems having the power to speak about a single individual, we abuse this notion and write $\cSTARSAT_\calQ(B)$ for the problems $\STARSAT_\calQ(B)$ without $\TSAT_\calQ(B)$.

\subsection{Both quantifiers}

\begin{theorem}[\cite{pr78,vawo86,doma00}]\label{thm:OCSAT_exall_in_EXPTIME}
$\OCSAT_\exall(\CloneBF)\in\EXPTIME$.
\end{theorem}

Due to the interreducibilities stated in \Cref{sec:prelims}, it suffices to 
show lower bounds for \ALCTSAT\ and upper bounds for \ALCOCSAT. 
Moreover \Cref{lem:Base_Independence} enables us to restrict the proofs to the standard basis of each clone for stating general results.

The following theorem improves \cite{MS10} by stating completeness results. 

\begin{theorem}\label{thm:ALCOCSAT_exall_results}
  Let $B$ be a finite set of Boolean operators.
  \begin{enumerate}
    \item If $\CloneI\subseteq[B]$ or $\CloneN_2\subseteq[B]$, then $\TSAT_\exall(B)$ is \EXPTIME-complete.
    \item If $\CloneI_0\subseteq[B]$ or $\CloneN_2\subseteq[B]$, then $\cSTARSAT_\exall(B)$ is \EXPTIME-complete.
    \item If $[B]\subseteq\CloneR_0$, then $\TSAT_\exall(B)$ is trivial.  
    \item If $[B]\subseteq\CloneR_1$, then $\STARSAT_\exall(B)$ is trivial.  
  \end{enumerate}
\end{theorem}

\begin{Proof}
Parts 1.--4. are formulated as  
\Cref{lem:ALCOCSAT_exall_E_V_EXPTIME-hard,lem:ALCOCSAT_exall_N2_EXPTIME-HARD,lem:ALCOCSAT_exall_R1_trivial,lem:ALCTSAT_exall_R0_trivial,lem:ALCOCSAT_exall_I0_EXPTIME-HARD},
and are proven below. 
\end{Proof}

Part (2) for $\CloneI_0$ generalizes the \EXPTIME-hardness of subsumption for \FLzero and \AL with respect to
GCIs \cite{gimcwiko02,don03,bbl05,hof05}. The contrast to the tractability of
subsumption with respect to GCIs in \EL, which uses only existential quantifiers,
undermines the observation that, for negation-free fragments, the choice of
the quantifier affects tractability and not the choice between conjunction and disjunction.
DL-Lite and \ALU cannot be put into this context because they use unqualified restrictions.

Parts (1) and (2) show that satisfiability with respect to theories is already intractable for even
smaller sets of Boolean operators. One reason is that sets of axioms already contain limited forms
of implication and conjunction. This also causes the results of this analysis to differ from similar analyses
for sub-Boolean modal logics in that hardness already holds for bases of clones that are comparatively
low in Post's lattice.

Part (3) reflects the fact that $\TSAT$ is less expressive than the other three decision problems:
it cannot speak about one single individual.
%


\begin{lemma}[\cite{MS10}]\label[lemma]{lem:ALCOCSAT_exall_R1_trivial}
  Let $B$ be a finite set of Boolean operators \st $B$ contains only $\top$-reproducing operators.
  Then $\ALCOCSAT_\exall(B)$ is trivial.
\end{lemma}


\begin{lemma}[\cite{MS10}]\label[lemma]{lem:ALCTSAT_exall_R0_trivial}
    Let $B$ be a finite set of Boolean operators \st $B$ contains only $\bot$-reproducing operators. Then $\TSAT_\exall(B)$ is trivial.
\end{lemma}

\begin{lemma}\label[lemma]{lem:ALCOCSAT_exall_E_V_EXPTIME-hard}
  Let $B$ be a finite set of Boolean operators with $\{\false,\dand\}\subseteq [B]$, or $\{\false,\dor\}\subseteq [B]$. Then $\cSTARSAT_\exall(B)$ is \EXPTIME-complete. If all self-dual operators can be expressed in $B$, then $\TSAT_\exall(B)$ is \EXPTIME-complete.
\end{lemma}
\begin{Proof}
The membership in \EXPTIME for $\OCSAT_\exall(B)$ follows from \Cref{thm:OCSAT_exall_in_EXPTIME} in combination with \Cref{lem:Base_Independence}.

For \EXPTIME-hardness, we first consider the case $\dand\in B$ and reduce from the positive entailment problem for Tarskian set constraints in \cite{gimcwiko02}: thus we start from the question if $\calT\models A\dsub B$, for concepts $A,B$ and a terminology $\calT$ that uses the quantifiers $\forall$ and $\exists$, and $\sqcap$ as the only Boolean connective.
Now $\calT$ just consists of concepts that contain $\dand$. 
Hence $\calT\models A\dsub B$ if and only if $\calT'\notin\TSAT_\exall(\{\sqcap,\true,\false\})$, for
$\calT':=\calT\cup\{\true\dsub\exists R.(A \sqcap B'),~B'\equiv\exists R_B.\top,~B\equiv\forall R_B.\false\}$,
where $B'$ is a new atomic concept and $R,R_B$ are new roles.
This holds as $A$ does not imply $B$ iff there is an instance of $A$ which is not an instance of $B$. As $B$ and $B'$ are declared disjoint, the claim applies.
Now for $\TCSAT_\exall(\{\false,\dand\})$, we transform $\calT'$ into $\calT''$ by substituting the two introduced occurrences of $\top$ with a fresh concept name $C$
and put $C$ into the instance of $\TCSAT_\exall(\{\false,\dand\})$ we are reducing to. Then, $\calT\models A\dsub B$
iff $(\calT'',C) \notin\TCSAT_\exall(\{\false,\dand\})$.

For $\TCSAT_\exall(\{\false,\dor\})$, we modify the above definition of $\calT''$ to dispose of the introduced conjunction:
using a fresh atomic concept $D$, we set
$\calT':=\calT\cup\{D\dsub A,~D\dsub B',~\true\dsub\exists R.D,~B'\equiv\exists R_B.C,~B\equiv\forall R_B.\false\}$.

  The remaining case for the self-dual operators follows from
    \Cref{lem:lewis-schnoor,lem:topbot-always-above-neg},
  as all self-dual functions in combination with the constants $\true,\false$ (to which we have access as $\lnot$ is self-dual) can express any arbitrary Boolean function.
\end{Proof}

\begin{lemma}\label[lemma]{lem:ALCOCSAT_exall_I0_EXPTIME-HARD}
  $\cSTARSAT_\exall(\{\false\})$ and $\TSAT_\exall(\{\false,\true\})$ are \EXPTIME-complete.
\end{lemma}
\begin{Proof}
  For the upper bound apply \Cref{thm:OCSAT_exall_in_EXPTIME} and \Cref{lem:Base_Independence}. For hardness, we reduce from $\TSAT_\exall(\{\dand,\false,\true\})$ to $\TSAT_\exall(\{\false,\true\})$---the former shown to be \EXPTIME-complete in the proof of \Cref{lem:ALCOCSAT_exall_E_V_EXPTIME-hard}. The main idea is an extension of the normalization rules in \cite{bra04b}. The following normalization rules have been stated and proven to be correct in \cite{bra04b}:
  \[
  \begin{array}{llcl}
    (\mathbf{NF1}) 
      &\hat C\dand D\axiomsign E
      &\rightsquigarrow
      &\{A\equiv\hat C,A\dand D\axiomsign E\}\\
    (\mathbf{NF2}) 
      &C\axiomsign D\dand \hat E
      &\rightsquigarrow
      &\{C\axiomsign D\dand A,A\equiv\hat E\}\\
    (\mathbf{NF3}) 
      &\exists r.\hat C\axiomsign D 
      &\rightsquigarrow 
      &\{A\equiv\hat C,\exists r.A\axiomsign D\}\\
    (\mathbf{NF4}) 
      &C\axiomsign\exists r.\hat D
      &\rightsquigarrow 
      &\{C\axiomsign\exists r.A,A\equiv\hat D\}\\
    (\mathbf{NF5}) 
      &C\dsub D\dand E
      &\rightsquigarrow 
      &\{C\dsub D,C\dsub E\}\\
    (\mathbf{NF6}) 
      &C\equiv D
      &\rightsquigarrow 
      &\{C\dsub D,D\dsub C\}
  \end{array}\]
  where $\axiomsign\in\{\dsub,\equiv\}$, $\hat C$ states that the concept description $C$ is no concept name, and $A$ is a new concept name.
  
  Now we want to extend these rules for conjunctions on the left side of GCIs and for $\forall$-quantification:
  \[
  \begin{array}{llcl}
    (\mathbf{NF3b}) 
      &\forall r.\hat C\axiomsign D
      &\rightsquigarrow
      &\{A\equiv\hat C,\forall r.A\axiomsign D\}\\
    (\mathbf{NF4b}) 
      &C\axiomsign\forall r.\hat D
      &\rightsquigarrow
      &\{A\equiv\hat D,C\axiomsign\forall r.A\}\\
    (\mathbf{NF7}) 
      &A\dand B\dsub C
      &\rightsquigarrow
      &\{A\dsub\exists R_A.\true,B\dsub\forall R_A.A',\exists R_A.A'\dsub C\}
  \end{array}
  \]
  where $R_A$ is a fresh role, and $A'$ is a fresh concept name. For ($\mathbf{NF7}$) we will prove its correctness.
  
  Assume $A\sqcap B\dsub C$ holds in the interpretation $\calI=(\Delta^\calI,\cdot^\calI)$. Thus for each individual $w\in\Delta^\calI$ with $w^\calI\supseteq\{A,B\}$ it holds $C\in w^\calI$ as assumed.
  
  In the following we will construct a modified interpretation $\calI'$ from $\calI$ that satisfies the axioms constructed by ($\mathbf{NF7}$), \ie, the axioms in $\{A\dsub\exists R_A.\true,B\dsub\forall R_A.A',\exists R_A.A'\dsub C\}$. As $A\in w^{\calI'}$, we add one $R_A$-edge to the same individual $w$, and due to $B\dsub\forall R_A.A'$ we must add $A'$ to $w^{\calI'}$. Finally the last GCI is satisfied as we have $C\in w^{\calI'}$.
  
  For the opposite direction assume $A\sqcap B\dsub C$ cannot be satisfied, \ie, in every interpretation there is an individual which is an instance of $A$ and $B$ but not of $C$. Hence we take an arbitrary interpretation $\calI$ such that it satisfies the first two axioms $A\dsub\exists R_A.\true$ and $B\dsub\forall R_A.A'$.
   Due to our assumption every individual $w$ is in instance of $A$ and $B$, and hence we have an $R_A$-edge to an individual where $A'$ must hold. Therefore the left side of the third axiom is fulfilled but $C$ does not hold for the individual $w$. Hence this axiom is not satisfied and we have the desired contradiction.

As this normalization procedure runs in polynomial time and eliminates every conjunction of concepts, we have a reduction from $\TCSAT_\exall(\{\dand,\false\})$ to $\TCSAT_\exall(\{\false\})$, and also from $\TSAT_\exall(\{\dand,\false,\true\})$ to $\TSAT_\exall(\{\true,\false\})$. Hence the Lemma applies.
\end{Proof}

\begin{lemma}\label[lemma]{lem:ALCOCSAT_exall_N2_EXPTIME-HARD}
 $\STARSAT_\exall(\CloneN_2)$ is \EXPTIME-complete.
\end{lemma}
\begin{Proof}
  The upper bound follows from \Cref{thm:OCSAT_exall_in_EXPTIME} and \Cref{lem:Base_Independence}. For the lower bound use 
 \Cref{lem:topbot-always-above-neg} to simulate $\true$ and $\false$ with fresh atomic concepts. Then the argumentation follows similarly to \Cref{lem:ALCOCSAT_exall_I0_EXPTIME-HARD,lem:ALCOCSAT_exall_E_V_EXPTIME-hard}.
\end{Proof}

\subsection{Restricted quantifiers}
In this section we investigate the complexity of the problems $\OCSAT_\calQ$, $\OSAT_\calQ$, $\TCSAT_\calQ$, and $\TSAT_\calQ$, where $\mathcal Q$
contains at most one of the quantifiers $\exists$ or $\forall$. 
Even the case $\calQ = \emptyset$ is nontrivial: for example, $\TSAT_\calQ(B)$ does not reduce to propositional satisfiability for $B$
because restricted use of implication and conjunction is implicit in sets of axioms.

\subsubsection{\TSAT-Results}
\begin{theorem}\label{thm:TSAT_emptyset}
Let $B$ be a finite set of Boolean operators.
\begin{enumerate}
  \item If $\CloneL_3\subseteq[B]$ or $\CloneM\subseteq[B]$, then $\TSAT_\emptyset(B)$ is \NP-complete.
  \item If $\CloneE=[B]$ or $\CloneV=[B]$, then $\TSAT_\emptyset(B)$ is \P-complete.
  \item If $[B]\in\{\CloneI,\CloneN_2,\CloneN\}$, then $\TSAT_\emptyset(B)$ is \NL-complete.
  \item Otherwise (if $[B]\subseteq\CloneR_1$ or $[B]\subseteq\CloneR_0$), then $\TSAT_\emptyset(B)$ is trivial.
\end{enumerate}
\end{theorem}

\begin{Proof}
\NP-completeness for (1) is composed of on the one hand the upper bound which results from $\OCSAT_\exists(\dand,\lnot,\true,\false)$ which is proven to be in \NP in \Cref{lem:OCSAT_emptyset_BF_membership} and on the other hand the lower bounds which are proven in \Cref{lem:TSAT_emptyset_D_hardness,lem:TSAT_emptyset_L3_hardness}.
Both lower bounds of (2) will be proven through \Cref{lem:TSAT_emptyset_E_hardness,lem:TSAT_emptyset_V_hardness}. The upper bound is due to $\OCSAT_\exists(\dand,\true,\false)$ which is shown to be in \P in \Cref{lem:OCSAT_exists_E_membership}.
The membership of the third item results from $\TCSAT_\emptyset(\lnot,\true)$ which is proven to be in \NL in \Cref{TCSAT_emptyset_N_membership} and the hardness result is proven in  \Cref{TSAT_emptyset_I_hardness}. 
Item (4) follows through \Cref{lem:ALCOCSAT_exall_R1_trivial,lem:ALCTSAT_exall_R0_trivial}.
\end{Proof}

\begin{lemma}\label[lemma]{lem:TSAT_emptyset_D_hardness}
Let $B$ be a set of Boolean operators s.t. all self-dual or monotone operators are in $[B]$. Then $\TSAT_\emptyset(B)$ is \NP-hard.
\end{lemma}
\begin{Proof}
We start with the implication problem for the self-dual (resp. monotone) fragment of propositional logic $\IMP(\CloneD)$ (resp. $\IMP(\CloneM)$), which is shown to be \co\NP-complete in \cite{bmtv09}. To establish \NP-hardness of $\TSAT_\emptyset(\CloneM)$, we reduce from the complement of $\IMP(\CloneM)$ in the following way. Let $\varphi,\psi$ be two propositional formulae with monotone operators only. Then
\begin{align*}
(\varphi,\psi)\notin\IMP(\CloneM) &\iff \varphi\not\models\psi \\
&\iff \exists \theta: \theta\models\varphi\land\lnot\psi\\
&\iff \{C_\psi\dsub\false,\true\dsub C_\varphi\}\in\TSAT_\emptyset(\CloneM),
\end{align*}
where $C_\varphi$ and $C_\psi$ are concepts corresponding to $\varphi,\psi$ in the usual way.

For $\TSAT_\emptyset(\CloneD)$, we use the same reduction, but need to replace the introduced operators $\true,\false$ as in Lemma \ref{lem:topbot-always-above-neg}.
\end{Proof}

\begin{lemma}\label[lemma]{lem:TSAT_emptyset_L3_hardness}
Let $B$ be a set of Boolean operators s.t. $\CloneL_3=[B]$, then $\TSAT_\emptyset(B)$ is \NP-hard.
\end{lemma}
\begin{Proof}
Here we will provide a reduction from the \NP-complete problem \OneInThreeSAT which is defined as follows: given a formula $\varphi=\bigwedge_{i=1}^n\bigvee_{j=1}^3l_{ij}$, where $l_{ij}$ are literals, we ask for the existence of a satisfying assignment which fulfills exact one literal per clause (\cite{sch78}). In the following we are allowed to use the binary exclusive-or as we have access to negation because $x\dxor x\dxor z\dxor\true\equiv\lnot z$, and we have access to both constants \true and \false due to \Cref{lem:topbot-always-above-neg}. Thus we are able to use the binary exclusive-or operator because $x\dxor y\dxor\true\dxor\true\equiv x\dxor y$.

The main idea of the reduction is to use for each clause $(x\lor y\lor z)\in\varphi$ an axiom $\true\dsub x\dxor y\dxor z$ to enforce that only one literal is satisfied. As for this axiom it is possible to have all literals satisfied we need some additional axioms to circumvent this problem.

Let $\varphi$ defined as above, then the reduction is defined as $\varphi\mapsto\calT$, where
\begin{align*}
  \calT &:= \set{\true\dsub f(l_{i1})\dxor f(l_{i2}) \dxor f(l_{i3}) \dxor s^i \dxor \true}{1\leq i\leq n}\cup\\
  &~~\cup\set{\true\dsub f(l_{i1})\dxor f(l_{i2}) \dxor f(l_{i3})}{1\leq i\leq n}\cup\\
  &~~\cup\set{s^i_1\dsub f(l_{i1})\dxor f(l_{i2})}{1\leq i\leq n}\cup\\
  &~~\cup\set{s^i_2\dsub f(l_{i1})\dxor f(l_{i3})}{1\leq i\leq n}\cup\\
  &~~\cup\set{s^i_3\dsub f(l_{i2})\dxor f(l_{i3})}{1\leq i\leq n}\cup\\
  &~~\cup\set{s^i\dsub s^i_1\dxor s^i_2\dxor s^i_3 }{1\leq i\leq n}\cup\\
  &~~\cup\set{\true\dsub A_x\dxor A_{x'}}{x\text{ variable in }\varphi},
\end{align*}
where $f(x)=A_x$ and $f(\bar x)= A_{x'}$. Now we claim that $\varphi\in\OneInThreeSAT$ iff $\calT\in\TSAT_\emptyset(\CloneL_0)$.

Consider an arbitrary clause $c=x\lor y \lor z$ from $\varphi$ with $x,y,z$ literals. Then following axioms which differ for convenience slightly from the notion above are part of $\calT$

\begin{minipage}[t]{.4\textwidth}
\begin{align}
\true &\dsub x\dxor y \dxor z\dxor s \dxor \true\\
\true &\dsub x\dxor y \dxor z\\\notag
s_1 &\dsub x\dxor y\\\notag
s_2 &\dsub x\dxor z\\\notag
s_3 &\dsub y\dxor z\\\notag
s&\dsub s_1\dxor s_2\dxor s_3.
\end{align}
\end{minipage}
\begin{minipage}[t]{.6\textwidth}
\[
\begin{array}{ccc|ccc|c|cc}
x&y&z&s_1&s_2&s_3&s&(1)&(2)\\\hline
0&0&0&   &   &   &   & &\no\\
0&0&1&\b0&1  &0  &1  &\yes&\yes\\
0&1&0&1  &\b0&0  &1  &\yes&\yes\\
0&1&1&   &   &   &   & &\no\\
1&0&0&1  &0  &\b0&1  &\yes&\yes\\
1&0&1&   &   &   &   & &\no\\
1&1&0&   &   &   &   & &\no\\
1&1&1&\b0&\b0&\b0&\b0&\no&\yes
\end{array}
\]
\end{minipage}\medskip

The table on the upper right shows each possible assignment for $x,y,z$ and suitable assignments for the $s_i$s and the validity of the axioms (1) and (2). Underlined numbers denote must set truth values enforced by the axioms whereas blank cells denote arbitrary choices. If at least one of (1) and (2) are contradicted then there exists no interpretation for $\calT$. At first we start with an interpretation that assigns the individuals $x,y,z$ to the recent world in some way. Then we immediately observe if axiom (2) is contradicted or not. If it is not contradicted then we have to look at the remaining $s_i$ axioms in order to find an extension of this interpretation which assigns the $s_i$s and $s$ in a way such that (2) is not violated whenever we have an interpretation which corresponds to a valid \OneInThreeSAT assignment. Otherwise we have to show that there exists no possible extension that falsely satisfies axiom (2).

Thus the table shows that for every eligible assignment we always have a fulfilling interpretation, and for ever improper assignment it is not possible to construct a fulfilling one. 

\end{Proof}

\begin{lemma}\label[lemma]{lem:TSAT_emptyset_E_hardness}
Let $B$ be a set of Boolean operators s.t. $\CloneE=[B]$, then $\TSAT_\emptyset(B)$ is \P-hard.
\end{lemma}
\begin{Proof}
In the following we will state a $\leqcd$-reduction from the complement of the \P-complete problem \HGAP, which is the accessibility problem for directed hypergraphs.
In a given hypergraph $H=(V,E)$, a hyperedge $e\in E$ is a set of source nodes $src(e)\subseteq V$ and one destination node $dest(e)\in V$. Instances of \HGAP consist of a directed hypergraph $H=(V,E)$, a set $S\subseteq V$ of source nodes, and a target node $t\in V$. Now the question is whether there exists a hyperpath
from the set $S$ to the node $t$, i.e., whether there are hyperedges $e_1,e_2,\dots,e_k$ s.t. for each $e_i$ there are $e_{i_1},\dots,e_{i_\nu}$ with $1\leq i_1,\dots,i_\nu<i$ and $\bigcup_{j\in\{i_1,\dots,i_\nu\}}dest(e_j)\supseteq src(e_i)$, and $src(e_1)=S$ and $dest(e_k)=t$.

\HGAP remains \P-complete even if we restrict the hyperedges to contain at most two source nodes \cite{sriy90}. W.l.o.g. assume that if there is a path from $S$ to $t$, then the last edge of that path is a usual edge with only one source node.

Let $G=(V,E)$ be a directed hypergraph, $\{s_1,\dots,s_k\}=S\subseteq V$ with $s_1,\dots,s_k\in V$ be the set of source nodes, and $t\in V$ be the target node. For each node $v\in V$, we use a new atomic concept $v$. In addition let $t,t'$ be fresh atomic concepts. 
Now define
\begin{multline*}
  \calT := \set{u_1\dand\dots\dand u_k\dsub v}{(u_1,\dots,u_k;v)\in E}\cup\\
  \cup\{\true\dsub s_1\dand\dots \dand s_k \dand t',t\dand t'\dsub\false\}.
\end{multline*}
Then $(G,S,t)\in\HGAP\iff\calT\notin\TSAT_\emptyset(\{\dand,\true,\false\})$.

``$\Rightarrow$'': Assume there is a hyperpath from $S$ to $t$ as above. Thus in every interpretation $\calI=(\Delta^\calI,\cdot^\calI)$ it holds for all $w\in\Delta^\calI$ that $s_1,\dots,s_k,t'\in w^\calI$. As the before mentioned hyperpath exists, $t$ must also be in $w^\calI$ through the chain of axioms that correspond to the hyperedges in the path. This violates the axiom $t\dand t'\dsub\false$.

``$\Leftarrow$'': Assume there is no hyperpath from $S$ to $t$ in $G=(V,E)$. Hence there is no chain of axioms that enforce $t$ to be true in every state. Therefore we are able to construct a satisfying interpretation in the following way: $\calI=(\{w\},\cdot^\calI)$ and
\begin{align*}
w^\calI := \set{v}{(s_1,\dots,s_k;v)\in E^*}\cup\{t'\},
\end{align*}
where $E$ is the transitive closure of $E$. Please note that $(s_1,\dots,s_k;t)\notin E^*$ and thus $t\notin w^\calI$. Therefore, all axioms are satisfied and $\calT\in\TSAT_\emptyset(\{\dand,\true,\false\})$.
\end{Proof}

\begin{lemma}\label[lemma]{lem:TSAT_emptyset_V_hardness}
Let $B$ be a set of Boolean functions s.t. $\CloneV=[B]$, then $\TSAT_\emptyset(B)$ is \P-hard.
\end{lemma}
\begin{Proof}
To realize the desired lower bound, we use \Cref{lem:contraposition} to state a reduction from $\TSAT_\emptyset(\CloneE)$ to $\TSAT_\emptyset(\CloneV)$.
\end{Proof}

\begin{lemma}\label[lemma]{TSAT_emptyset_I_hardness}
Let $B$ be a set of Boolean functions \st $\CloneI=[B]$, then $\TSAT_\emptyset(B)$ is \NL-hard.
\end{lemma}
\begin{Proof}
  For proving \NL-hardness we will reduce from the complement of the graph accessibility problem \GAP which is \NL-complete. Consider a given directed graph $G=(V,E)$ and two nodes $s,t\in V$ as the recent instance for \GAP asking for a path from $s$ to $t$ in $G$.
        We introduce a concept name $A_v$ per node $v \in V$ and define
  \begin{align*}
    \calT := \set{A_u\dsub A_v}{(u,v)\in E}\cup\{\true\dsub A_s, A_t\dsub\false\}.
  \end{align*}
  We will now prove that $(G,s,t)\notin\GAP\iff\calT\in\TSAT_\emptyset(B)$.
  
  ``$(G,s,t)\notin\GAP\Rightarrow\calT\in\TSAT_\emptyset(B)$'': Assume there is no path from $s$ to $t$.
        Take the interpretation $\calI:=(\{x\},\cdot^\calI)$ with 
        \[A^\calI_v:=
        \begin{cases}
      \{x\} & \text{ if $v$ is reachable from $s$},\\
      \emptyset & \text{ otherwise},
    \end{cases}
        \]
        for each $v\in V$. Then $A^\calI_t=\emptyset$ and with that all axioms are satisfied. 
        Thus it holds that $\calI\models\calT$.
  
  ``$(G,s,t)\in\GAP\Rightarrow\calT\notin\TSAT_\emptyset(B)$'': Now assume we have a path $\pi=v_1,\dots,v_k$ in $G$ with $k\in\N$, $(v_i,v_{i+1})\in E$, $v_i\in V$ for $1\leq i\leq k$, $v_1=s$, and $v_k=t$ from $s$ to $t$. Now any interpretation needs to include an individual $x$  instantiating $A_s$ (else $\true\dsub A_s$ would be contradicted) and also $A_{v_2},\dots, A_{v_k}=A_t$. But with $A_t\in x^\calI$ we contradict the axiom $A_t\dsub\false$. Thus $\calI\not\models\calT$, and with that $\calT\notin\TSAT_\emptyset(B)$.
\end{Proof}

\begin{lemma}\label[lemma]{TSAT_emptyset_I_membership}
Let $B$ be a set of Boolean functions \st $\CloneI=[B]$, then $\TSAT_\emptyset(B)$ is in \NL.
\end{lemma}
\begin{Proof}
The main idea is to do a path search in a concept dependence graph---a reduction to the complement of \GAP. A given $\calT$ is mapped to $G=(V,E)$ where
\begin{align*}
  V &:= \set{v_A,v_B}{A\dsub B\in \calT}\cup\{v_\true,v_\false\}\text{ and }\\
  E &:= \set{(v_A,v_B)}{A\dsub B}.
\end{align*}
Now it holds $\calT\in\TSAT_\emptyset(B)\iff (G,v_\true,v_\false)\notin\GAP$. Please note that we need to add $v_\true,v_\false$ to $V$ in order to keep consistency if at least one of $\true$ and $\false$ is not part of an axiom side. If $\calT$ is not satisfiable, then in every interpretation there is at least one axiom contradicted. \Wlog the contradicted axiom is of the form $C\dsub \false$ and $C$ is instantiated by some individual $x$. Thus there must be a chain of axioms that enforce $C$ to be true and it can be easily shown that this chain starts at some axiom $\true\dsub C'$. Hence we have a path starting at $v_\true$ in the Graph $G$ which leads to a node $v_\false$. For the opposite direction the argumentation is analogue.
\end{Proof}

\begin{theorem}\label{thm:TSAT_onequantifier}
Let $B$ be a finite set of Boolean operators and $\calQ\in\{\forall,\exists\}$.
\begin{enumerate}
  \item If $\CloneM\subseteq[B]$ or $\CloneN_2\subseteq[B]$, then $\TSAT_\calQ(B)$ is \EXPTIME-complete.
  \item If $\CloneE=[B]$, $\CloneV=[B]$, or $\CloneI=[B]$, then $\TSAT_\calQ(B)$ is \P-complete.
  \item Otherwise (if $[B]\subseteq\CloneR_1$ or $[B]\subseteq\CloneR_0$), then $\TSAT_\calQ(B)$ is trivial.
\end{enumerate}\end{theorem}

\begin{proof}
For the monotone case in (1) consider \Cref{lem:TSAT_exists_M_hardness,lem:TSAT_forall_M_hardness}.
The proof for $\CloneN_2$ can be found in \Cref{lem:TSAT_exists_N2_hardness}. The respective upper bounds for (1) result from \Cref{thm:OCSAT_exall_in_EXPTIME} in combination with \Cref{lem:Base_Independence}.
The needed lower bound for the \P-hardness results in (2) is shown for $\TSAT_\exists(\CloneI)$ in \Cref{lem:TSAT_exists_I_hardness} (case $\forall$ is due to \Cref{lem:contraposition}). The membership in \P for the cases in (3) result on the one hand from $\OCSAT_\exists(\dand,\true,\false)$ which is shown to be in \P in \Cref{lem:OCSAT_exists_E_membership} and on the other hand from $\TSAT_\forall(\dand,\true,\false)$ is proven in \Cref{lem:TSAT_forall_E_membership}. The two remaining upper bounds for $[B]=\CloneV$ follow from the complementary problem through \Cref{lem:contraposition}.

Item (3) follows through \Cref{lem:ALCOCSAT_exall_R1_trivial,lem:ALCTSAT_exall_R0_trivial}.
\end{proof}

\par\smallskip\noindent
Part (3) generalizes the fact that every \EL- and \FLzero-TBox is satisfiable, and the whole theorem shows that
separating either conjunction and disjunction, or the constants is the only way to achieve tractability for \TSAT.

\begin{lemma}\label[lemma]{lem:TSAT_exists_M_hardness}
Let $B$ be a set of Boolean functions s.t. $\CloneM=[B]$, then $\TSAT_\exists(B)$ is \EXPTIME-hard.
\end{lemma}
\begin{Proof}
For \EXPTIME-hardness, we will reduce from the complement of the subsumption problem \wrt TBoxes for the logic \ELU, which has been investigated in \cite[Thm. 7]{bbl05}.
\ELU is \ALC restricted to the operators $\top,\dand,\dor,\exists$. 
Now it holds that
\begin{align*}
  &(\calT,A,B)\in\ELU\text{-}\SUBS\\
  &~~~\iff \calT \models A\dsub B                                                                                \\
  &~~~\iff \text{for all } \calI: \calI\models \calT \text{ implies } \calI\models A\dsub B                      \\
  &~~~\iff \text{there is no }\calI:\calI\models\calT \text{ and } \calI\models A\not\dsub B                     \\
  &~~~\iff \text{there is no }\calI:\calI\models\calT\text{ and } \calI\models \true\dsub\exists R.A\dand\lnot B \\
  &~~~\iff \text{there is no }\calI:\calI\models\underbrace{\calT\cup\big\{\true\dsub\exists R.(A\dand B'),\true\dsub B\dor B',B\dand B'\dsub \false\big\}}_{\calT'} \\
  &~~~\iff \calT'\notin\TSAT_\exists(\CloneM),
\end{align*}
for a fresh role $R$ and a fresh concept $B'$.
\end{Proof}

\begin{lemma}\label[lemma]{lem:TSAT_forall_M_hardness}
Let $B$ be a set of Boolean functions s.t. $\CloneM=[B]$, then $\TSAT_\forall(B)$ is \EXPTIME-hard.
\end{lemma}
\begin{Proof}
As in the proof of \Cref{lem:TSAT_emptyset_V_hardness}, we can reduce from the dual problem $\TSAT_\exists(B)$ through \Cref{lem:contraposition}.
\end{Proof}

\begin{lemma}\label[lemma]{lem:TSAT_exists_N2_hardness}\label[lemma]{lem:TSAT_forall_N2_hardness}
Let $B$ be a set of Boolean functions s.t. $\CloneN_2=[B]$ and $\calQ\in\{\forall,\exists\}$, then $\TSAT_\calQ(B)$ is \EXPTIME-hard.
\end{lemma}
\begin{Proof}
We reduce from $\TSAT_\exall(\CloneI)$, which is shown to be \EXPTIME-complete in \Cref{lem:ALCOCSAT_exall_N2_EXPTIME-HARD}.
As known from \Cref{lem:topbot-always-above-neg}, we can simulate the constants using new concept names and negation. Additionally observe that,
although \calQ contains only one quantifier, the other quantifier can be expressed using $\lnot$.
\end{Proof}

\begin{lemma}\label[lemma]{lem:TSAT_forall_E_membership}
Let $B$ be a set of Boolean functions s.t. $\CloneE=[B]$, then $\TSAT_\forall(B)$ is in \P.
\end{lemma}
\begin{Proof}
Here we will specify an algorithm for satisfiability similar to the one in \cite{bra04} that constructs iteratively the transitive closure of atomic concepts that imply each other. Thus, informal speaking, starting by the empty set $S_0:=\emptyset$, for each $S_i$ we look at each axiom $C\dsub D$ and add $D$ to $S_{i+1}$ iff $C\in S_i$. 
The construction of these sets is defined inductively as follows, where $\calT$ is a TBox that is in normalform (\ie, $\calT$ contains only expressions of the form $C\dsub D$, $C_1\dand C_2\dsub D$, $\forall r.C\dsub D$, or $C\dsub\forall r.D$, where $C$ and $D$ are atomic concepts and $r$ is a role--please note that for each $S_i$ it holds $S_i\subseteq(\CONC\cup\{\true,\false\})^*$):
\begin{description}
  \item[\textbf{(IS1)}] If $C_1\in S_i(C)$ and $C_1\dsub D\in\calT$, then $S_{i+1}(C):= S_i(C)\cup\{D\}$.
  \item[\textbf{(IS2)}] If $C_1,C_2\in S_i(C)$ and $C_1\dand C_2\dsub D\in\calT$, then $S_{i+1}(C):=S_i(C)\cup\{D\}$.
  \item[\textbf{(IS3)}] If $C_1\in S_i(C)$ and $C_1\dsub\forall r.D\in\calT$ and $D_1\in S_i(D)$ and $\forall r.D_1\dsub C\in\calT$, then $S_{i+1}(C):= S_i(C)\cup\{D\}$.
\end{description}
The construction for each of those sets $S_i$ takes  time at most $\calO(|\calT|)$ and eventually stops for an atomic concept $C$ if $S_i(C)=S_{i+1}(C)$ for some $i\in\N$.

We now claim that $\calT\in\TSAT_\forall(B)$ iff $\false\notin S_*^\calT(\true)$, where $S_*^\calT(\true)$ denotes the transitive closure of $S_i$ for $\true$ \wrt \calT.

``$\Rightarrow$'': Let $\calT\in\TSAT_\forall(B)$ via the interpretation \calI. Hence $\calI\models\calT$ and in particular for each $C\dsub D\in\calT$ it holds that $C^\calI\subseteq D^\calI$. As \textbf{(IS1)} to \textbf{(IS3)} hold, we have $\false\notin S_*^\calT(\true)$, otherwise there exist $C_1\dsub D_1,\dots, C_\ell\dsub D_\ell\in\calT$ \st $C_1=\true$ and $D_\ell=\false$, and $C_1$ implies $D_\ell$ through these axioms. We show this by induction on $n$, where $n$ is the index of the first $S_i$ with $\false\notin S_i^\calT(\true)$.

Let $n=1$, then $C_1=\true$ and $D_1=\false$; hence we apply \textbf{(IS1)} for $\true\dsub\false\in\calT$ and $\false\in S_1^\calT(\true)$.

$n\to n+1$: Let $1\leq i,j\leq n$,
\begin{enumerate}
  \item $C_{n+1}= D_j,$ and $D_j\in S_n(\true)$, then $D_{n+1}\in S_{n+1}(\true)$.
  \item $C_{n+1}= D_i\dand D_j,$ and $D_i,D_j\in S_n(\true)$, then $D_{n+1}\in S_{n+1}(\true)$.
  \item $C_k= D_j$, $1\leq k\neq j< n$, $D_k=\forall r.C_s$, $k\leq s \leq n$, and $C_i\in S_n(\true)$, and $\forall r. C_i\dsub D_n\in\calT$, then $D_{n+1}\in S_{n+1}(\true)$. 
\end{enumerate}
Hence, if $D_{n+1}=\false$, then $\false\in S_{n+1}^\calT(\true)$.

The argumentation for the opposite direction is analogue to \cite{bra04c}.
\end{Proof}

\begin{lemma}\label[lemma]{lem:TSAT_exists_I_hardness}
 Let $B$ be a set of Boolean functions s.t. $\CloneI=[B]$, then $\TSAT_\exists(B)$ is \P-hard.
\end{lemma}
\begin{Proof}
We will reduce the word problem for the Turing machine model that characterizes \LOGCFL to $\SUBS_\exists(\emptyset)$.
Together with the trivial reduction $\SUBS_\exists(\emptyset) \leqslant \overline{\TSAT_\exists(\CloneI)}$,
justified by $(\calT,A,B) \in \SUBS_\exists(\emptyset)$ iff $(\calT \cup \{\top \equiv A, B \equiv \bot\}) \notin \TSAT_\exists(\CloneI_0)$,
this will provide \LOGCFL-hardness of $\TSAT_\exists(\CloneI)$. Observe that \LOGCFL is closed under complement \cite{bcd+89}. As in the proof the runtime of the Turing machine is not relevant we achieve instead a \P-hardness result (because an \NL-Turing machine with arbitrary runtime leads to the class \P \cite{co71b}).

Let $M$ be a nondeterministic Turing machine, which has access to a read-only input tape,
a read-write work tape and a stack, and whose runtime is bounded by a polynomial in the size of the input.
Let $M$ be the 6-tuple $(\Sigma, \Psi, \Gamma, Q, f, q_0)$, where
\begin{itemize}
  \item
    $\Sigma$ is the input alphabet;
  \item
    $\Psi$ is the work alphabet containing the empty-cell symbol $\#$;
  \item
    $\Gamma$ is the stack alphabet containing the bottom-of-stack symbol $\Box$;
  \item
    $Q$ is the set of states;
  \item
    $f \,:\, Q \times \Sigma \times \Psi \times \Gamma ~\to~ Q \times \Psi \times \{-,+\}^2 \times (\Gamma\setminus\{\Box\})^\star$
    is the state transition function which describes a transition where the machine is in a state, reads an input symbol,
    reads a work symbol and takes a symbol from the stack, and goes into another state, writes a symbol to the work tape,
    makes a step on each tape (left or right) and possibly adds a sequence of symbols to the stack;
  \item
    $q_0 \in Q$ is the initial state.
\end{itemize}
We assume that each computation of $M$ starts in $q_0$ with the heads at the left-most position of each tape
and with exactly the symbol $\Box$ on the stack. W.l.o.g., the machine accepts whenever the stack is empty, regardless
of its current state.

Let $x = x_1\dots x_n$ be an input of $M$.
We consider the configurations that can occur during any computation of $M(x)$ in two versions.
A \emph{shallow configuration} of $M(x)$ is a sequence $(p\delta_1\dots\delta_{k-1}q\delta_k\dots\delta_\ell)$,
where
\begin{itemize}
  \item
    $p \in \{1,\dots,n\}$ is the current position on the input tape, represented in binary;
  \item
    $\ell \in O(\log n)$ is the maximal number of positions on the work tape of $M$ relevant for the computations of $M(x)$;
  \item
    $\delta_1,\dots,\delta_\ell$ is the current content of the work tape;
  \item
    $k$ is the current position on the work tape;
  \item
    $q$ is the current state of $M$.
\end{itemize}
The initial shallow configuration $(0q_0\#\dots\#)$ is denoted by $S_0$.
\newcommand{\sconf}[1]{\ensuremath{\mathcal{SC}_{#1}}}%
\newcommand{\dconf}[1]{\ensuremath{\mathcal{DC}_{#1}}}%
Let $\sconf{M,x}$ be the set of all possible shallow configurations that can occur during any computation of $M(x)$.
The cardinality of this set is bounded by a polynomial in $n$ because the number of work-tape cells used is logarithmic in $n$ and the binary counter for the position on the input tape is logarithmic in $n$.

A \emph{deep configuration} of $M(x)$ is a sequence $(R_1\dots R_mp\delta_1\dots\delta_{k-1}q\delta_k\dots\delta_\ell)$,
where the $R_i$ are the symbols currently on the stack and the remaining components are as above.
Let $\dconf{M,x}$ be the set of all possible deep configurations that can occur during any computation of $M(x)$.
The cardinality of this set can be exponential as soon as $\Gamma$ has more than two elements besides $\Box$.
This is not a problem for our reduction, which will only touch shallow configurations.

We now construct an instance of $\SUBS_\exists(\emptyset)$ from $M$ and $x$. We use each shallow configuration $S \in \sconf{M,x}$
as a concept name and each stack symbol as a role name. The TBox $\calT_{M,x}$ describes all possible computations of
$M(x)$ by containing an axiom for every two deep configurations that the machine can take on before and after
some computation step. A deep configuration $D$ is represented by the concept corresponding to $D$'s
shallow part, preceded by the sequence of existentially quantified stack symbols corresponding to
the stack content in $D$. The TBox $\calT_{M,x}$ is constructed from a set of axioms
per entry in $f$. (We will omit the subscript from now on.) For the instruction
\[
  (q,\sigma,\delta,R) \mapsto (q',\delta',-,-,R_1\dots R_k)
\]
of $f$, we add the axioms
\begin{align}
  \exists R. & (\bin(p)\delta_0\dots\delta_{i-1}q\delta\delta_{i+1}\dots\delta_\ell) ~\sqsubseteq \notag\\
  \exists R_1\dots\exists R_k. & (\bin(p\!\dotdiv\!1)\delta_0\dots\delta_{i-2}q'\delta_{i-1}\delta'\delta_{i+1}\dots\delta_\ell) \label{eq:f-axiom}
\end{align}
for every $p$ with $x_p = \sigma$, every $i = 1,\dots,\ell$, and all $\delta_0,\dots,\delta_{i-1},\delta_{i+1},\dots,\delta_\ell$.
The expression $p\dotdiv1$ stands for $p-1$ if $p \geqslant 2$ and for $1$ otherwise, reflecting the assumption that the machine does not move on the input tape
on a ``go left'' instruction if it is already on the left-most input symbol. This behaviour can always be assumed w.l.o.g.
In case $k = 0$, the quantifier prefix on the right-hand side is empty. For instructions of $f$ requiring ``$+$'' steps on any of the tapes,
the construction is analogue. The number of axioms generated by each instruction is bounded by the number of shallow configurations; therefore
the overall number of axioms is bounded by a polynomial in $n \cdot |f|$.

Furthermore, we use a fresh concept name $B$ and add an axiom $\calS \sqsubseteq B$ for each shallow configuration $\calS$. Also we add a single axiom $S\dsub\exists\Box.S_0$ to \calT. 
The instance of $\SUBS_\exists(\emptyset)$ is constructed as $(\calT,~ S,~ B)$. 
\calT 
can be constructed in logarithmic space.
It remains to prove the following claim.

\smallskip\noindent\emph{Claim.}
$M(x)$ has an accepting computation if and only if $S \sqsubseteq_\calT B$.

\smallskip\noindent\emph{Proof of Claim.}
For the ``$\Rightarrow$'' direction, we observe that, for each step in the accepting computation,
the (arbitrary) concept associated with the pre-configuration is subsumed by the concept associated
with the post-configuration. More precisely, if $M(x)$ makes a step
\[
  (q,\sigma,\delta,R) \mapsto (q',\delta',-,-,R_1\dots R_k),
\]
then its deep configuration before that step has to be
\begin{align*}
  S_1\dots S_jR & p\delta_0\dots\delta_{i-1}q\delta\delta_{i+1}\dots\delta_\ell,
  \intertext{for some $S_1,\dots,S_j\in\Gamma$, $\delta_0,\dots,\delta_{i-1},\delta_{i+1},\dots,\delta_\ell\in\Psi$ and $p\in\N$., and the deep configuration after that step is}
  S_1\dots S_jR_1\dots R_k & (p\!\dotdiv\!1)\delta_0\dots\delta_{i-2}q'\delta_{i-1}\delta'\delta_{i+1}\dots\delta_\ell.
\end{align*}
The set of axioms constructed in \ref{eq:f-axiom} ensures that there is an axiom that implies
\begin{align*}
  \exists S_1\dots \exists S_j.\exists R. & (\bin(p)\delta_0\dots\delta_{i-1}q\delta\delta_{i+1}\dots\delta_\ell)
  ~\sqsubseteq_\calT \\
  \exists S_1\dots \exists S_j.\exists R_1\dots R_k.& (\bin(p\!\dotdiv\!1)\delta_0\dots\delta_{i-2}q'\delta_{i-1}\delta'\delta_{i+1}\dots\delta_\ell).
\end{align*}
Since some computation of $M(x)$ reaches a configuration with an empty stack,
we can conclude that some atomic concept corresponding to a shallow configuration $\calS$,
and therefore also $B$, subsumes $\exists\Box.S_0$ which subsumes $S$ (per definition).

For the ``$\Leftarrow$'' direction, we assume that $M(x)$ has no accepting computation.
This means that, during every computation of $M(x)$, the stack does never become empty.
From the set of all computations of $M(x)$, we will show that there exists an interpretation \calI that satisfies \calT,
but not $S \sqsubseteq B$; hereby we can conclude $(\calT,S,B)\notin\SUBS_\exists(\emptyset)$. 

Observe that any atomic concept besides $S$ and $B$ in \calT correspond to a specific shallow  configuration of $M(x)$. Let $T_{M(x)}:=(V,E)$ denote the computation tree of $M(x)$. Thus every node $v\in V$ represents a deep configuration of $M(x)$ which will be denoted via $C_v$.
Then for two nodes $u,v\in V$ with $(u,v)\in E$ it holds that $C_u\vdash_{M}C_v$. In the following we will describe how to construct an interpretation \calI from $T_{M(x)}$ which has a witness for $S^{\calI}\not\subseteq  B^{\calI}$. Further on we will denote individuals $\mathbf x$ in bold font to differ them from the input $x$ for $M$. For ease of notion we will write for some shallow configuration $\mu\in\sconf{M,x}$ in the following also $\mu$ for the respecting concept in \calT.

The root of $T_{M(x)}$ is the initial configuration 
$\Box0q_0\underbrace{\#\dots\#}_{\ell}.$
Now we will define $\calI(S):=\bigcup_{i\ge0}\calI_i(S)$ starting with $\Delta^{\calI_0(S)}:=\{\mathbf x\}$ and
\begin{itemize}
  \item $S^{\calI_0(S)}:=\{\mathbf x\}$, and
  \item $y\in(S_0)^{\calI_0(S)}$ with $(\mathbf x,\mathbf y)\in\Box^{\calI_0(S)}$ (\ie, $(\exists\Box.S_0)^{\calI_0(S)}=\{\mathbf x\}$)
\end{itemize}
inductively as follows. $\mathbf{(1)}$ For every node $v\in V$ s.t. $C_v=S_1\dots S_jR\mu$ with $\mu\in \bin(\N)\times\Psi^h\cdot Q\cdot\Psi^k$ and $h+k=\ell-1$ is the corresponding configuration in $M(x)$ and let $\mathbf x_1,\dots,\mathbf x_j,\mathbf x_r,\mathbf x_\mu\in\Delta^{\calI_i(S)}$ be individuals s.t. $(\mathbf x_1,\mathbf x_2)\in (S_1)^{\calI_i(S)}, (\mathbf x_2,\mathbf x_3)\in(S_2)^{\calI_i(S)},\dots, (\mathbf x_j,\mathbf x_r)\in(S_j)^{\calI_i(S)},(\mathbf x_r,\mathbf x_\mu)\in R^{\calI_i(S)}$ and $\mathbf x_\mu\in\mu^{\calI_i(S)}$: 

if $u\in V$ with $(v,u)\in E$ is a post configuration $C_u = S_1\dots S_jR_1\dots R_k\lambda$ for $\lambda\in \bin(\N)\times\Psi^h\cdot Q\cdot\Psi^k$ and $h+k=\ell-1$ of the configuration  $C_v$ in the computation of $M(x)$, \ie, $C_v\vdash_MC_u$, then 
\begin{itemize}
  \item add $\mathbf x_r$ to $\lambda^{\calI_{i+1}(S)}$ for $k=0$, and otherwise
  \item if there do not exist $\mathbf y_1,\dots,\mathbf y_k\in\Delta^{\calI_i(S)}$ with $(\mathbf x_r,\mathbf y_1)\in (R_1)^{\calI_i(S)},(\mathbf y_1,\mathbf y_2)\in(R_2)^{\calI_i(S)},\dots,(\mathbf y_{k-1},\mathbf y_k)\in(R_k)^{\calI_i(S)}$ and $\mathbf y_k\in\lambda^{\calI_i(S)}$, then introduce new individuals $\mathbf y_1,\dots,\mathbf y_k$ to $\Delta^{\calI_{i+1}(S)}$ and add $(\mathbf x_\mu,\mathbf y_1)$ to $(R_1)^{\calI_{i+1}(S)}$, $(\mathbf y_1,\mathbf y_2)$ to $(R_2)^{\calI_{i+1}(S)}$, $\dots$, $(\mathbf y_{k-1},\mathbf y_k)$ to $(R_k)^{\calI_{i+1}(S)}$ and include $\mathbf y_k$ into $\lambda^{\calI_{i+1}(S)}$.
\end{itemize}
$\mathbf{(2)}$ For every individual $\mathbf x\in\Delta^{\calI_i(S)}$ and deep configuration $\chi$ that is also a shallow configuration with $\mathbf x\in\chi^{\calI_{i}(S)}$ include $\mathbf x$ into $B^{\calI_{i+1}(S)}$.

In the following we will show that $\calI(S)$ is indeed a valid interpretation for $\calT$ but $S\not\dsub_\calT B$. As there is no axiom in \calT with $S$ on the right side it holds that $|S^{\calI(S)}|=1$. Assume there is some GCI $G=A_G\dsub B_G\in\calT$ which is violated in $\calI(S)$, \ie, we have some individual $\mathbf x'\in\Delta^{\calI(S)}$ s.t. $\mathbf x'\in(A_G)^{\calI(S)}$ but $\mathbf x'\notin(B_G)^{\calI(S)}$. As in \calT there are two different kinds of axioms we have to distinguish these cases (because the axiom with $S$ on the left side cannot be such a violated axiom):
\begin{enumerate}
  \item If $G=\alpha\dsub\beta\in\calT$ for $\alpha$ and $\beta$ being atomic (this is the case for axioms with concepts representing shallow configurations on the left side and $B$ on the right side), then $\mathbf x'\in\alpha^{\calI(S)}$ but $\mathbf x'\notin\alpha^{\calI(S)}$. Now consider the least index $n$ s.t. $\mathbf x'\in\alpha^{\calI_n(S)}$. As $\alpha$ represents clearly a shallow configuration and $\beta=B$ then $\mathbf x'$ is added to $\beta^{\calI_{n+1}(S)}\subseteq\beta^{\calI(S)}$ by $\mathbf{(2)}$, which contradicts the assumption.
  \item If $G=\exists R.\mu\dsub \exists R_1.\dots \exists R_k.\lambda\in\calT$ wherefore exist some entry in $f$ from $M$ s.t. $(S_1\dots S_jR\mu)\vdash_M(S_1\dots S_jR_1\dots R_k\lambda)$ for some stack symbols $S_1,\dots,S_j$, then $\mathbf x'\in(\exists R.\mu)^{\calI(S)}$ but $\mathbf x'\notin (\exists R_1.\dots \exists R_k.\lambda)^{\calI(S)}$. Now let $n$ denote the least index s.t. $\mathbf y$ is added to $(\mu)^{\calI_n(S)}$ and there must be some $m<n$ s.t. $(\mathbf x',\mathbf y)$ is added to $R^{\calI_m(S)}$. Then in step $\mathbf{(1)}$ there are $\mathbf y_1,\dots,\mathbf y_{k}$ added to $\Delta^{\calI_{n+1}(S)}$, the corresponding $R_i$-edges are added to their respective $(R_i)^{\calI_{n+1}(S)}$-set and $\mathbf y_k$ is added to $\lambda^{\calI_{n+1}(S)}$ obeying $\mathbf x\in(\exists R_1.\dots \exists R_k.\lambda)^{\calI_{n+1}(S)}\subseteq(\exists R_1.\dots \exists R_k.\lambda)^{\calI(S)}$. This contradicts our assumption again.
\end{enumerate}
Consequently $\calI(S)$ is a model of $\calT$. Now assume that $S^{\calI(S)}\subseteq B^{\calI(S)}$. Thus for the starting point $\mathbf x$ which is added to $S^{\calI(S)}$ at the initial construction step of $\calI(S)$, it holds in particular that $\mathbf x\in B^{\calI(S)}$. As $\mathbf x$ is added to $B^{\calI(S)}$ if and only if $\mathbf x$ is added to $\mu^{\calI(S)}$ for some shallow configuration $\mu$, we can conclude that an accepting configuration must be reachable in $T_{M(x)}$ which contradicts our assumption (of the absence of such a computation sequence).
Thus an inductive argument proves that $\mu\in \mathbf x^{\calI_n(S)}$ for $\{\mathbf x\}=S^{\calI(S)}$ implies that $M$ reaches an accepting configuration on $x$ in $T_{M(x)}$. \smallskip

\emph{Claim.} Let $C=(R_1\dots R_k\mu)$ be a configuration. It holds for all $n\in\N$ that if $\mathbf x\in (\exists R_1.\dots\exists R_k.\mu)^{\calI_n(S)}$ and $\{\mathbf x\}= S^{\calI(S)}$ then $M$ reaches $C$ in the computation on $x$ in its computation tree $T_{M(x)}$.

\emph{Induction basis.} Let $n=1$ and $C=(R_1\dots R_k.\mu)$ for $\mu\in\sconf{M,s}$ be some configuration with $\mathbf x\in(\exists R_1.\dots\exists R_k.\mu)^{\calI_1(S)}$ and $\{\mathbf x\}=S^{\calI(S)}$. Thus the individual $\mathbf x$ is added to $(\exists R_1.\dots\exists R_k.\mu)^{\calI_1(S)}$ because we have some axiom s.t. $\exists\Box.(\bin(0)\#\dots\#)\dsub\exists R_1.\dots\exists R_k.\mu\in\calT$ as we only have one step in this case. Hence $C$ can be reached from the initial configuration $\Box0q_0\#\dots\#$ in one step via the transition that corresponds to the before mentioned axiom, \ie, $\Box0q_0\#\dots\#\vdash_MR_1.\dots R_k\mu$.

\emph{Induction step.} Let $n>1$ and assume the claim holds for all $m<n$. Now we have some configuration $C=(S_1\dots S_jR_1\dots R_k\mu)$ for $\mu\in\sconf{M,x}$ with $\mathbf x\in (\exists S_1.\dots\exists S_j.\exists R_1.\dots\exists R_k.\mu)^{\calI_{n}(S)}$ and $\{\mathbf x\}=S^{\calI(S)}$. By induction hypothesis we have some other configuration $C'=(S_1\dots S_jR\lambda)$ with $\lambda\in\sconf{M,x}$ from which $C$ occurs in one step, \ie, $C'\vdash_M C$, and $C$ is reachable on the computation of $M(x)$ and $\mathbf x\in (\exists S_1.\dots\exists S_j.\exists R.\lambda)^{\calI_{n-1}(S)}$. Thus we have also some axiom that adds $\mathbf x$ to $(\exists S_1.\dots\exists S_j.\exists R_1.\dots\exists R_k.\mu)^{\calI_n(S)}$ in $\mathbf{(1)}$. This axiom is of the form $\exists R.\lambda\dsub\exists R_1.\dots\exists R_k.\mu\in\calT$. As $M$ reaches $C'$ by induction hypothesis and $C$ can be reached via one step from $C'$ and $\mathbf x$ is an instance of $\exists S_1.\dots\exists S_j.\exists R_1.\dots\exists R_k.\mu$, $M$ can also reach $C$ within the computation on $x$.

Hence this contradicts our assumption that $M$ does not accept $x$ and completes our proof. 
\end{Proof}

\subsubsection{\TCSAT-, \OSAT-, \OCSAT-Results.}
\label{sec:cstarsat_restricted}

\begin{theorem}\label{thm:cSAT_emptyset}
Let $B$ be a finite set of Boolean operators.
\begin{enumerate}
  \item If $\CloneS_{11}\subseteq[B]$ or $\CloneL_3\subseteq[B]$ or $\CloneL_0\subseteq[B]$, then $\cSTARSAT_\emptyset(B)$ is \NP-complete.
  \item If $[B]\in\{\CloneE_0,\CloneE,\CloneV_0,\CloneV\}$, then $\cSTARSAT_\emptyset(B)$ is \P-complete.
  \item If $[B]\in\{\CloneI_0,\CloneI,\CloneN_2,\CloneN\}$, then $\cSTARSAT_\emptyset(B)$ is \NL-complete.
  \item Otherwise (if $[B]\subseteq\CloneR_1$), then $\cSTARSAT_\emptyset(B)$ is trivial.
\end{enumerate}
\end{theorem}

\begin{Proof}
  \NP-hardness for (1) follows from the respective $\TSAT_\emptyset(B)$ results in  \Cref{lem:TSAT_emptyset_D_hardness,lem:TSAT_emptyset_L3_hardness} in combination with \Cref{lem:Lewis_Trick2} for the lower bound. The membership in \NP is shown in \Cref{lem:OCSAT_emptyset_BF_membership}.
  
  The lower bounds for (2) result from $\TSAT_\emptyset(\dand,\true,\false)$ and $\TSAT_\emptyset(\dor,\true,\false)$ shown in \Cref{lem:TSAT_emptyset_V_hardness,lem:TSAT_emptyset_E_hardness} in combination with \Cref{lem:Lewis_Trick2} while the upper bound applies due to $\OCSAT_\exists(\dand,\true,\false)$ which is proven to be in \P in \Cref{lem:OCSAT_exists_E_membership}.

  The lower bound of (3) is proven in \Cref{TCSAT_emptyset_I0_hardness}. The upper bound follows from \Cref{lem:OCSAT_emptyset_N_membership,TCSAT_emptyset_N_membership}.
  
  (4) is due to \Cref{lem:ALCOCSAT_exall_R1_trivial}.
\end{Proof}

\begin{lemma}\label[lemma]{lem:OCSAT_emptyset_BF_membership}
Let $B$ be a set of Boolean functions s.t. $[B]\subseteq\CloneBF$. Then $\OCSAT_\emptyset(B)$ is in \NP.
\end{lemma}
\begin{Proof}
We will reduce $\OCSAT_\emptyset(B)$ to \SAT, the satisfiability problem for propositional formulae.
Due to \Cref{lem:Base_Independence}, we can assume that $B = \{\dand,\neg\}$.
Let $\big((\calT,\calA),C\big)$ be an instance of $\OCSAT_\emptyset(B)$.
Since $\ALC_\emptyset(B)$ does not have quantifiers, \calT only makes propositional statements about all individuals
and cannot enforce more individuals than those in \calA. Let $D_j \sqsubseteq E_j$, $j=1,\dots,n$, be the axioms in \calT
and $a_1,\dots,a_m$ the individuals occurring in \calA. We introduce a fresh atomic proposition $p^i_A$
for each $i=0,\dots,m$ and each atomic concept $A$ occurring in $(\calT,\calA)$.
Every $p^i_A$ expresses that $A$ has as instance either the individual $a_i$ (if $i\geq1$) or an
an instance of $C$ (if $i=0$). Although $C$ may have several instances, the absence of
quantifiers allows us to identify them with a single individual.

For $i=0,\dots,m$, we define a function $f^i$ that maps from arbitrary concepts occurring in $\big((\calT,\calA),C\big)$
to propositional formulae as follows:
\begin{align*}
  f^i(A)            & = p^i_A \qquad \text{for atomic concepts }A, \\
  f^i(\true)        & = \ONE,  \qquad f^i(\false) = \ZERO,         \\
  f^i(\neg A)       & = \overline{f^i(A)},                         \\
  f^i(A_1\dand A_2) & = f^i(A_1) \land f^i(A_2).
\end{align*}
We express the instance $\big((\calT,\calA),C\big)$ using the following propositional formulae:
\begin{align*}
  \varphi_{\calT}         & = \bigwedge_{i=0}^m\;\bigwedge_{j=1}^n\big(f^i(D_j) \to f^i(E_j)\big) \displaybreak[0],\\
  \varphi_{\calA}         & = \bigwedge_{i=1}^m\;\bigwedge_{D(a_i)\in\calA}f^i(D)                 \displaybreak[0],\\
  \varphi_C               & = f^0(C)                                                              \displaybreak[0],\\[4pt]
  \varphi_{\calT,\calA,C} & = \varphi_{\calT} \land \varphi_{\calA} \land\varphi_C.
\end{align*}
We will now show that $\big((\calT,\calA),C\big) \in \OCSAT_\emptyset(B)$ if and only if $\varphi_{\calT,\calA,C} \in \SAT$.

For ``$\Rightarrow$'', assume that $\big((\calT,\calA),C\big) \in \OCSAT_\emptyset(B)$. Then there is an interpretation
\calI such that $\calI \models (\calT,\calA)$ and $C^{\calI} \neq \emptyset$. Fix individuals $x_0,\dots,x_m \in \Delta^{\calI}$
such that $x_0 \in C^{\calI}$ and $x_i = a_i^{\calI}$ for $i=1,\dots,m$. Now construct a propositional assignment $\beta$
such that $\beta(p^i_A) = 1$ if and only if $x_i \in A^{\calI}$. It is straightforward to show by induction on $X$
that for every, possibly complex, concept $X$ occurring in $\big((\calT,\calA),C\big)$ and each $i=0,\dots,m$, it holds that
$\beta\big(f^i(X)\big) = 1$ if and only if $x_i \in X^{\calI}$. Using this equivalence, we show that $\beta(\varphi_{\calT,\calA,C}) = 1$.
\begin{itemize}
  \item
    $\beta(\varphi_{\calT}) = 1$ because, for every $i,j$, the axiom $D_j \sqsubseteq E_j$ in \calT
    ensures that $x_i \in D_j^{\calI}$ implies $x_i \in E_j^{\calI}$.
  \item
    $\beta(\varphi_{\calA}) = 1$ because every $D(a_i)$ in \calA means that $x_i \in D^{\calI}$.
  \item
    $\beta(\varphi_C) = 1$ because $x_0 \in C^{\calI}$.
\end{itemize}

For ``$\Leftarrow$'', assume that $\varphi_{\calT,\calA,C} \in \SAT$. Then there is an assignment $\beta$ under which all three conjuncts
$\varphi_{\calT},\varphi_{\calA},\varphi_C$ evaluate to 1. We construct an interpretation \calI from $\beta$ as follows.
$\Delta^{\calI} = \{x_0,\dots,x_m\}$; for every $i=0,\dots,m$, every individual $a$ in \calA and every atomic concept
$A$ in $\big((\calT,\calA),C\big)$: $a_i^{\calI} = x_i$ and $x_i \in A^{\calI}$ if and only if
$\beta(p^i_A) = 1$. As above, it is straightforward to show that $\beta\big(f^i(X)\big) = 1$ if and only if $x_i \in X^{\calI}$,
for every $X$ in $\big((\calT,\calA),C\big)$ and every $i=0,\dots,m$. Using this equivalence, we show that
$\calI \models (\calT,\calA)$ and $C^{\calI} \neq \emptyset$.
\begin{itemize}
  \item
    $\calI \models D_j \sqsubseteq E_j$, $j = 1,\dots,n$ because, for every $i=0,\dots,m$, the conjuncts in
    $\varphi_{\calT}$ ensure that $\beta\big(f^i(D_j)\big) = 1$ implies that $\beta\big(f^i(E_j)\big) = 1$,
    and therefore $x_i \in D_j^{\calI}$ implies $x_i \in E_j^{\calI}$.
  \item
    $\calI \models D(a_i)$, $D(a_i) \in \calA$, because the conjuncts in $\varphi_{\calA}$ ensure that $x_i \in D^{\calI}$.
  \item
    $C^{\calI} \neq \emptyset$ because $\varphi_C$ ensures that $x_0 \in C^{\calI}$.
\end{itemize}
\end{Proof}

\begin{lemma}\label[lemma]{TCSAT_emptyset_N_membership}
Let $B$ be a set of Boolean functions \st $\CloneN=[B]$, then $\TCSAT_\emptyset(B)$ is in \NL.
\end{lemma}
\begin{Proof}
Here we will provide a nondeterministic algorithm for $\TSAT_\emptyset(B)$ that runs in logarithmic space, which can be generalized to also work with $\TCSAT_\emptyset(B)$ instances $(\calT,C)$ by adding an axiom $\true\dsub C$ to the input terminology (in our case this maintains satisfiability because we can only talk about one individual).
The algorithm consists of a search for cycles with contradictory atomic concepts in the (directed) implication graph $G_\calT$ which is induced by \calT. 

W.l.o.g. assume $\calT$ to be normalized in a way that all blocks of leading negations $\lnot$ in front of concepts are replaced by one negation if the number was odd, and completely removed otherwise. Thus $\calT$ consists only of axioms $C\dsub D$, where $C,D$ are atomic concepts, constants, or its negations. 
The before mentioned implication graph $G_\calT=(V,E)$ is constructed from $\calT$ as follows:
\begin{align*}
V &:=\set{v_A,v_{\lnot A}}{A\text{ is an atomic concept in }\calT}\cup\{v_\true,v_\false\},\\
E &:=\set{(v_C,v_D)}{C\dsub D\in\calT}\cup\\
  &~~~~\cup\set{(v_\false,v_A),(v_A,v_\true)}{A\text{ is an atomic concept in }\calT}\cup\{(v_\false,v_\true)\}.
\end{align*}

Now we claim that $\calT\in\TSAT_\emptyset(B)$ iff $G_\calT$ does not contain a cycle that contains both nodes $v_A, v_{\lnot A}$ for some $A\in\CONC\cup\{\true,\false\}$.

"$\Rightarrow$": Let $\calT\in\TSAT_\emptyset(B)$ witnessed by the interpretation $\calI=(\Delta^\calI,\cdot^\calI)$. W.l.o.g. assume $\Delta^\calI=\{x\}$ by the same argumentation as in \Cref{lem:OCSAT_emptyset_BF_membership}. Then it holds that $\calI\models\calT$. Hence each axiom is satisfied, and with that there is no axiom $C\dsub D$ s.t. $x\in C^\calI$ but $x\notin D^\calI$. Now assume that we have a cyclic path $\pi$ containing the nodes $v_A$ and $v_{\lnot A}$. If $x\in A^\calI$ then for all successor nodes $v_{A_1},v_{A_2},\dots$ of $v_A$ on $\pi$ it must hold that $x\in A_i^\calI$ for $i=1,2,\dots$, which is a contradiction to $\lnot A$ for which $v_{\lnot A}$ is a successor of $v_A$. If $x\notin A^\calI$ then $x\in(\lnot A)^\calI$. Thus for all axioms $A_1,A_2,\dots$ with $v_{A_1},v_{A_2},\dots$ being successor nodes of $v_{\lnot A}$ it must hold that $x\in(A_i)^\calI$. In particular this must hold for $v_A$ which is a contradiction to $x\notin A^\calI$.

"$\Leftarrow$": Assume that for each atomic concept $A$ (including \true and \false) there is no cyclic path containing $v_A$ and $v_{\lnot A}$. In the following we will construct an interpretation $\calI=(\{x\},\cdot^\calI)$ that satisfies \calT. For each concept $A\in\Conc(\{\true,\false,\lnot\})$ s.t. $\true\dsub^* A$, add $x$ to $A^\calI$. As we have $(v_A,v_{\mathord{\sim}A})\notin E^*$ (where $E^*$ is the transitive closure of $E$, and $\mathord{\sim} A=\lnot B$ if $A=B$ and $\mathord\sim A=B$ if $A=\lnot B$) it must hold that also $A\not\dsub^*\mathord\sim A$ and thus $\calI\models\calT$, as all remaining concepts are not enforced to be true. This completes the proof of the claim.\bigskip

The \NL-algorithm just checks for each concept $A$ that there is no cycle from $v_A$ containing $v_{\sim A}$.
\end{Proof}

\begin{lemma}\label[lemma]{TCSAT_emptyset_I0_hardness}
Let $B$ be a set of Boolean functions \st $\CloneI_0=[B]$, then $\TCSAT_\emptyset(B)$ is \NL-hard.
\end{lemma}
\begin{Proof}
This result directly follows from \Cref{TSAT_emptyset_I_hardness} in combination with \Cref{lem:Lewis_Trick2}.
\end{Proof}

\begin{lemma}\label[lemma]{lem:OCSAT_emptyset_N_membership}
Let $B$ be a finite set of Boolean operators s.t. $\CloneN=[B]$, then $\OCSAT_\emptyset(B)$ is in \NL.
\end{lemma}
\begin{Proof}
Let $B$ be a set of Boolean operators s.t. $\CloneN=[B]$.
The algorithm first checks whether the given TBox is solely satisfiable. Afterwards we need to ensure the given ABox is consistent together with the TBox. Therefore observe for an ABox $\calA$ the following property holds: $(\calA,\calT,C)\in\OCSAT_\emptyset(B)$ iff $(\calA\cup\{R(a,b)\},\calT,C)\in\OCSAT_\emptyset(B)$ for new individuals $a,b$ and a role $R$, as role assertions cannot affect the satisfiability of an instance if quantifiers are not allowed.
The algorithm now tests consecutively for each individual $a\in\calA$ if $(\calT^a,C)\in\TCSAT_\emptyset(B)$, where $\calT^a=\calT\cup\set{\true\dsub D}{D(a)\in\calA}$.

Now it holds that $(\calA,\calT,C)\in\OCSAT_\emptyset(B)$ iff $(\calT^a,C)\in\TCSAT_\emptyset(B)$ for all individuals $a\in\calA$ and $(\calT,C)\in\TCSAT_\emptyset(B)$.

If $\calI=(\Delta^\calI,\cdot^\calI)$ is an interpretation with $\calI\models(\calT,\calA)$ and $C^\calI\neq\emptyset$, then for the terminologies $\calT^a$ for each individual $a\in\calA$ it holds that $\calI|_{a}\models\calT^a$, where $\calI|_{a}$ is the restriction of \calI to the individual $a$. For the opposite direction to be considered, we have interpretations $\calI^a=(\Delta^{\calI^a},\cdot^{\calI^a})$ s.t. $\calI^a\models\calT^a$ and $C^{\calI^a}\neq\emptyset$. W.l.o.g. assume $\Delta^{\calI^a}=\{a\}$, then an easy inductive argument proves that $\calI\models(\calT,\calA)$ and $C^\calI\neq\emptyset$ for $\calI=(\bigcup_{a\in\calA}\Delta^{\calI^a},\cdot^{\bigcup_{a\in\calA}\calI^a})$.

This connection between $\OCSAT_\emptyset(B)$ and $\TCSAT_\emptyset(B)$ is possible as we can assume different individuals to be distinct. As besides of that point we cannot speak about more than one individual for a given TBox which is restricted to a single individual $a$, and therefore we may assume the concept $D$ to hold (and consider also the axiom $\true\dsub D$) if $D(a)\in\calA$ for $\calT^a$. 
\end{Proof}

\begin{theorem}\label{thm:starSAT_onequantifier}
Let $B$ be a finite set of Boolean operators, and $\calQ\in\{\forall,\exists\}$.
\begin{enumerate}
  \item If $\CloneS_{11}\subseteq[B]$, $\CloneN_2\subseteq[B]$, or $\CloneL_0\subseteq[B]$ then $\cSTARSAT_\calQ(B)$ is \EXPTIME-complete. 
  \item If $\CloneI_0\subseteq[B]\subseteq\CloneV$, then $\TCSAT_\exists(B)$ and $\cSTARSAT_\forall(B)$ are \P-complete\footnote{$\OSAT_\exists(B)$ and $\OCSAT_\exists(B)$ are \P-hard for $[B]\in\{\CloneV_0,\CloneV\}$ and in $\EXPTIME$.}.
  \item If $[B]\in\{\CloneE_0,\CloneE\}$, then $\cSTARSAT_\forall(B)$ is \EXPTIME-complete,\\
  and $\cSTARSAT_\exists(B)$ is \P-complete.
  \item If $[B]\subseteq\CloneR_1$, then $\cSTARSAT_\calQ(B)$ is trivial.
\end{enumerate}
\end{theorem}
\begin{Proof}
For (1) combine the \EXPTIME-completeness of $\TSAT_\calQ(\CloneM)$ shown in \Cref{lem:TSAT_exists_M_hardness} with the usual $\true$-knack known from \Cref{lem:Lewis_Trick2}.

The lower bound for $\CloneN_2$ is due to \Cref{lem:TSAT_forall_N2_hardness} to state a reduction from $\TSAT_\calQ(\CloneL)$ with \Cref{lem:Lewis_Trick2} to $\TCSAT_\calQ(\CloneL_0)$ for $\calQ\in\{\exists,\forall\}$.

The \EXPTIME-completeness in case (3) follows from \Cref{lem:TCSAT_forall_E0_hardness}. 
For the \P-complete cases in (2) and (3) the results are organized as follows:
\begin{itemize}
	\item the \P-hardness of these cases results from $\TSAT_\calQ(\true,\false)$ in \Cref{lem:TSAT_exists_I_hardness} in combination with \Cref{lem:Lewis_Trick2},
	\item the membership in \P of $\TCSAT_\forall(\dor,\true,\false)$ follows by $\OCSAT_\forall(\dor,\true,\false)$ in \Cref{lem:OCSAT_forall_V_membership},
	\item the membership in \P of $\TCSAT_\exists(\dor,\true,\false)$ follows by $\TSAT_\exists(\dor,\true,\false)$ in combination with \Cref{lem:TCSAT_reduces_to_TSAT_with_true},
	\item the membership in \P of $\TCSAT_\exists(\dand,\true,\false)$ follows by $\OCSAT_\exists(\dand,\true,\false)$ in \Cref{lem:OCSAT_exists_E_membership}.
\end{itemize}
(4) is due to \Cref{lem:ALCOCSAT_exall_R1_trivial}.
\end{Proof}

\par\smallskip\noindent
Theorem \ref{thm:starSAT_onequantifier} shows one reason why the logics in the \EL family have been much more successful
as ``small'' logics with efficient reasoning methods than the \FL family:
the combination of the $\forall$ with conjunction is intractable, while $\exists$ and conjunction are still in polynomial time.
Again, separating either conjunction and disjunction, or the constants is crucial for tractability.

\begin{lemma}\label[lemma]{lem:TCSAT_forall_E0_hardness}
Let $B$ be a finite set of Boolean operators s.t. $\CloneE_0=[B]$, then $\TCSAT_\forall(B)$ is \EXPTIME-hard.
\end{lemma}
\begin{Proof}
As a result from \cite{bbl05,hof05} the subsumption problem \wrt a TBox for the logic \FLzero (the description logic with $\forall$ and $\dand$ as allowed operators) is \EXPTIME-complete. For this lemma we will reduce from this problem in \FLzero. Observe that the following holds
\begin{align*}
  &(\calT, C,D)\in\FLzero\text{-}\SUBS\\ 
  &~~~\iff
    \forall \calI:\calI\models\calT\text{ it holds }\calI\models C\dsub D\\
    &~~~\iff\text{not}(\exists\calI:\calI\models\calT \text{ and } (C\dand \lnot D)^\calI\neq\emptyset)\\
    &~~~\iff\text{not}(\exists\calI:\calI\models\calT\cup\{D\dand D'\dsub\false\}\text{ and } (C\dand D')^\calI\neq\emptyset)\\
    &~~~\iff (\calT\cup\{D\dand D'\dsub\false\},C\dand D')\notin\TCSAT_\forall(B)
\end{align*}
\end{Proof}

\begin{lemma}\label[lemma]{lem:OCSAT_exists_E_membership}
Let $B$ be a finite set of Boolean operators s.t. $\CloneE=[B]$, then $\OCSAT_\exists(B)$ is in \P.
\end{lemma}
\begin{Proof}
To provide an algorithm running in polynomial time, we will reduce the given problem to the complement of the subsumption problem for the logic \ELplusplus, which is known to be \P-complete by \cite{bbl08}.

The reduction works as follows:
\begin{align*}
((\calT,\calA),C)\in\OCSAT_\exists(B) &\iff \exists \calI: \calT\models\calT \text{ and } \calC_\calA^\calI\neq\emptyset\text{ and } C^\calI\neq\emptyset\\
&\iff \exists \calI': \calI'\models\calT\cup\{\true\dsub\exists R.\calC_\calA\} \text{ and } C^{\calI'}\neq\emptyset\\
&\iff\calT\cup\{\true\dsub\exists R.\calC_\calA\}\not\models C\dsub\false\\
&\iff (\calT\cup\{\true\dsub\exists R.\calC_\calA\},C,\false)\notin\ELplusplus\text{-}\SUBS,
\end{align*}
where \calT is a TBox, \calA is an ABox, $R$ is a fresh role, and
$$
\calC_\calA := \bigsqcap_{C(a)\in\calA}\exists u.(\{a\}\dand C)\dand \bigsqcap_{r(a,b)\in\calA}\exists u.(\{a\}\dand\exists r.\{b\})
$$
is the concept constructed as in \cite{bbl05rep} from the ABox \calA, where $u$ is a fresh role name, and $\{a\}$ and $\{b\}$ denote nominals corresponding to the ABox individuals $a$ and $b$.
\end{Proof}

\begin{lemma}\label[lemma]{lem:OCSAT_forall_V_membership}
Let $B$ be a finite set of Boolean operators s.t. $\CloneV=[B]$, then $\OCSAT_\forall(B)$ is in \P.
\end{lemma}
\begin{Proof}
Here we use the result from \Cref{lem:OCSAT_exists_E_membership} and reduce to the dual problem $\OCSAT_\exists(\CloneE)$. Consider an ontology $(\calT,\calA)$ where $\calT$ is a TBox and $\calA$ an ABox, and a concept $C$ as the given instance of $\OCSAT_\forall(B)$. \Wlog assume $C$ to be atomic. Now first construct the new terminology $\calT'$ similarly to \Cref{lem:TSAT_emptyset_V_hardness}. Then add for each $A\in\CONC$ and hence each $A'$ the GCIs $A\dand A'\dsub\false$ to ensure they are disjoint. Denote this change by the terminology $\calT''$. Then it holds $((\calT,\calA),C)\in\OCSAT_\forall(B)\iff((\calT'',\calA),C')\in\OCSAT_\exists(\CloneE)$.
\end{Proof}

\par\medskip\noindent
\Cref{tbl:overview} gives an overview of our results. \Cref{ssec:overview} shows how the results arrange in
Post's lattice.
\begin{table}
	\[
	\begin{array}{l|c|c|c|c|c|c|c}
		\TSAT_\calQ(B) & \CloneI & \CloneV & \CloneE & \CloneN/\CloneN_2 & \CloneM & \CloneL_3\text{ to }\CloneBF & \text{otherwise}\\\hline\hline
		\calQ = \emptyset & \NL &\multicolumn{2}{c|}{\P} & \NL & \multicolumn{2}{c|}{\NP}&\text{trivial}\\\hline
		|\calQ|=1 & \multicolumn{3}{c|}{\P} & \multicolumn{3}{c|}{\EXPTIME}&\text{trivial}\\\hline
		\calQ=\{\exists,\forall\} & \multicolumn{6}{c|}{\EXPTIME}&\text{trivial}\\\hline
		\multicolumn{7}{c}{}
		\\
		\cSTARSAT_\calQ(B) & \CloneI/\CloneI_0 & \CloneV/\CloneV_0& \CloneE/\CloneE_0 & \CloneN/\CloneN_2 & \CloneS_{11}\text{ to }\CloneM & \CloneL_3/\CloneL_0\text{ to }\CloneBF &\text{otherwise}\\\hline\hline
		\calQ=\emptyset & \NL &\multicolumn{2}{c|}{\P} & \NL & \multicolumn{2}{c|}{\NP}&\text{trivial}\\\hline
		\calQ=\{\exists\}& \P & \P^\mathsection & \P &\multicolumn{3}{c|}{\EXPTIME}&\text{trivial}\\\hline
		\calQ=\{\forall\}& \multicolumn{2}{c|}{\P} & \multicolumn{4}{c|}{\EXPTIME}&\text{trivial}\\\hline
		\calQ=\{\exists,\forall\} & \multicolumn{6}{c|}{\EXPTIME}&\text{trivial}\\\hline
	\end{array}
	\]
	\caption{Complexity overview for all Boolean function and quantifier fragments. All results are completeness results for the given complexity class, except for the case marked \textsection: here, $\OCSAT$ and $\OSAT$ are in $\EXPTIME$ and $\P$-hard.}\label{tbl:overview}
\end{table}

\section{Conclusion}
With \Cref{thm:ALCOCSAT_exall_results,thm:TSAT_emptyset,thm:TSAT_onequantifier,thm:cSAT_emptyset,thm:starSAT_onequantifier},
we have completely classified the satisfiability problems connected to arbitrary terminologies and concepts for \ALC fragments obtained by arbitrary sets of Boolean operators and quantifiers---only the fragments emerging around ontologies with existential quantifier and disjunction as only allowed connective resisted a full classification. In particular we improved and finished the study of \cite{MS10}. In more detail we achieved a dichotomy for all problems using both quantifiers  (\EXPTIME-complete vs. trivial fragments), a trichotomy when only one quantifier is allowed (trivial, \EXPTIME-, and \P-complete fragments), and a quartering for no allowed quantifiers ranging from trivial, \NL-complete, \P-complete, and \NP-complete fragments.

Furthermore the connection to well-known logic fragments of \ALC, \eg, \FL and \EL
now enriches the landscape of complexity by a generalization of these results.
These improve the overall understanding of where the tractability border lies.
The most important lesson learnt is that the separation of quantifiers
together with the separation of either conjunction and disjunction,
or the constants, is the only way to achieve tractability in our setting.

Especially in contrast to similar analyses of logics using \emph{Post's lattice}, this study shows intractable fragments quite at the bottom of the lattice. This illustrates how expressive the concept of terminologies and assertional boxes is: restricted to only the Boolean function \emph{false} besides both quantifiers we are still able to encode \EXPTIME-hard problems into the decision problems that have a TBox and a concept as input. Thus perhaps the strongest source of intractability can be found in the fact that unrestricted theories already express limited implication and disjunction, and not in the set of allowed Boolean functions alone.

For future work, it would be interesting to see whether the picture changes if the use of general axioms is restricted,
for example to cyclic terminologies---theories where axioms are cycle-free definitions $A \equiv C$ with $A$ being atomic.
Theories so restricted are sufficient for establishing taxonomies.
Concept satisfiability for \ALC w.r.t\ acyclic terminologies is still \PSPACE-complete \cite{lut99}.
Is the tractability border the same under this restriction?
One could also look at fragments with unqualified quantifiers, \eg, $\ALU$ or the DL-lite family,
which are not covered by the current analysis.
Furthermore, since the standard
reasoning tasks are not always interreducible under restricted Boolean operators,
a similar classification for other decision problems such as concept subsumption is pending.


\subsection{Overview of the Results}\label{ssec:overview}
Regarding the number of possible fragments of the investigated decision problems by restricting the use of quantifiers and Boolean functions one would formally deduce the number of emerging fragments is infinite (as there are infinitely many different Boolean functions). Fortunately Post's lattice hides this infinity at two parts in the lattice, namely, the $c$-separating functions of degree $n$ and the clones around them. This is visualized by dashed lines in the lattice. To overcome this problem one tries to achieve the same upper and lower bounds for the clones above and below these infinite chains. Thus there are still all visualized nodes in the lattice remaining to get classified. Each of these clones induces a new decision problem parameterized by itself. Thus we have to deal with 54 relevant clones which means, all in all, $4\cdot54$ parameterized versions for all four decision problems. 

Therefore the next table will help to clarify the overall picture in the following way. Each row deals with the quantifier fragments whereas each column corresponds to one clone in the lattice. Here, we mostly used only the clones which are needed to state best upper and lower bounds. A cell in this table shows the complexity of this fragment (by name and color), wherefrom the lower and wherefrom the upper bound is applied or in which lemma the corresponding proof can be found. The "Lewis Knack" is proven in \Cref{lem:Lewis_Trick2}.

\newcommand{\white}[1]{{\color{white}\vspace{-1mm}#1}}
\begin{landscape}
\pagestyle{empty}
\setlength{\voffset}{+41mm}

{\tiny
\begin{center}
  \tablefirsthead{%
    \hline
    &
    $\CloneI_0$&
    $\CloneI$&
    $\CloneN_2$&
    $\CloneV_0$&
    $\CloneV$&  
    $\CloneE_0$&
    $\CloneE$&
    $\CloneS_{11}$&
    $\CloneD$&
    $\CloneM$&
    $\CloneR_1$&
    $\CloneR_0$\\
    \hline%
  }
  \tablehead{%
    \hline
    \multicolumn{13}{|l|}{\small\sl continued from previous page}\\
    \hline
    &
    $\CloneI_0$&
    $\CloneI$&
    $\CloneN_2$&
    $\CloneV_0$&
    $\CloneV$&  
    $\CloneE_0$&
    $\CloneE$&
    $\CloneS_{11}$&
    $\CloneD$&
    $\CloneM$&
    $\CloneR_1$&
    $\CloneR_0$\\
    \hline%
  }
  \tabletail{%
    \hline%
    \multicolumn{13}{|r|}{\small\sl continued on next page}\\
    \hline%
  }
  \tablelasttail{\hline}
  \bottomcaption{Complexity Overview, LB: Lower Bound, UP: Upper Bound, LK: Lewis Knack, con: contraposition}
  \begin{supertabular}{|c|p{2cm}|p{2.1cm}|p{2.1cm}|p{2.1cm}|p{2.1cm}|p{2.1cm}|p{2.1cm}|p{2.1cm}|p{2.1cm}|p{2.1cm}|p{1.6cm}|p{2.1cm}|}
    $\TSAT_\emptyset$&
      \cellcolor{lightgray}
      trivial,\newline
        $\TSAT_\exall(\CloneR_0)$ &
      \cellcolor{green}
      \NL-complete,\newline 
        LB: \Cref{TSAT_emptyset_I_hardness},\newline 
        UB: \Cref{TSAT_emptyset_I_membership} &
      \cellcolor{green}
      \NL-complete,\newline
        LB: $\TSAT_\emptyset(\CloneI)$,\newline
        UB: $\TCSAT_\emptyset(\CloneN)$
      &
      \cellcolor{lightgray}
      trivial, \newline
        $\TSAT_\exall(\CloneR_0)$
      &
      \cellcolor{blue}\white{%
      \P-complete,\newline 
        LB: \Cref{lem:TSAT_emptyset_V_hardness},\newline 
        UB: $\OCSAT_\forall(\CloneV)$}
      &
      \cellcolor{lightgray}
      trivial,\newline
        $\TSAT_\exall(\CloneR_0)$ 
      &
      \cellcolor{blue}\white{%
      \P-complete,\newline 
        LB: \Cref{lem:TSAT_emptyset_E_hardness},\newline 
        UB: $\OCSAT_\exists(\CloneE)$} 
      &
      \cellcolor{lightgray}
      trivial,\newline
        $\TSAT_\exall(\CloneR_0)$ 
      &
      \cellcolor{orange}
      \NP-complete,\newline 
        LB: \Cref{lem:TSAT_emptyset_D_hardness},\newline 
        UB: \Cref{lem:OCSAT_emptyset_BF_membership}
      &
      \cellcolor{orange}
      \NP-complete,\newline 
        LB: \Cref{lem:TSAT_emptyset_D_hardness},\newline 
        UB: \Cref{lem:OCSAT_emptyset_BF_membership}
      &
      \cellcolor{lightgray}
      trivial,\newline
        $\OCSAT_\exall(\CloneR_1)$ &
      \cellcolor{lightgray}
      trivial,\newline 
        $\TSAT_\exall(\CloneR_0)$
      \\\hline
    $\TSAT_\forall$&
      \cellcolor{lightgray}
      trivial,\newline 
        $\TSAT_\exall(\CloneR_0)$ &
      \cellcolor{blue}\white{
      \P-complete,\newline 
        LB: $\TSAT_\exists(\CloneI)$+con,\newline
        UB: $\TSAT_\exists(\CloneI)$+con}
      &
      \cellcolor{red}
      \EXPTIME-compl.,\newline
        LB: \Cref{lem:TSAT_forall_N2_hardness}
      &
      \cellcolor{lightgray}
      trivial,\newline 
        $\TSAT_\exall(\CloneR_0)$ 
      &
      \cellcolor{blue}\white{%
      \P-complete,\newline 
        LB: $\TSAT_\emptyset(\CloneV)$,\newline 
        UB: $\TSAT_\exists(\CloneE)$+con}
      &
      \cellcolor{lightgray}
      trivial,\newline
        $\TSAT_\exall(\CloneR_0)$ &
      \cellcolor{blue}\white{%
      \P-complete,\newline 
        LB: $\TSAT_\emptyset(\CloneE)$,\newline 
        UB: \Cref{lem:TSAT_forall_E_membership}} &
      \cellcolor{lightgray}
      trivial,\newline
        $\TSAT_\exall(\CloneR_0)$ &
      \cellcolor{red}
      \EXPTIME-compl.,\newline
        LB: $\TSAT_\forall(\CloneN_2)$
      &
      \cellcolor{red}
      \EXPTIME-compl.,\newline 
        LB: \Cref{lem:TSAT_forall_M_hardness}
      &
      \cellcolor{lightgray}
      trivial,\newline
        $\OCSAT_\exall(\CloneR_1)$ &
      \cellcolor{lightgray}
      trivial,\newline 
        $\TSAT_\exall(\CloneR_0)$
      \\\hline
    $\TSAT_\exists$&
      \cellcolor{lightgray}
      trivial,\newline 
        $\TSAT_\exall(\CloneR_0)$ &
      \cellcolor{blue}\white{
      \P-complete,\newline 
        LB: \Cref{lem:TSAT_exists_I_hardness},\newline
        UB: $\TCSAT_\exists(\CloneE)$}
      &
      \cellcolor{red}
      \EXPTIME-compl.,\newline
        LB: \Cref{lem:TSAT_exists_N2_hardness}
      &
      \cellcolor{lightgray}
      trivial,\newline 
        $\TSAT_\exall(\CloneR_0)$ &
      \cellcolor{blue}\white{%
      \P-complete,\newline 
        LB: $\TSAT_\emptyset(\CloneV)$,\newline 
        UB: $\TSAT_\forall(\CloneE)$+con}
      &
      \cellcolor{lightgray}
      trivial,\newline
        $\TSAT_\exall(\CloneR_0)$ &
      \cellcolor{blue}\white{%
      \P-complete,\newline 
        LB: $\TSAT_\emptyset(\CloneE)$,\newline 
        UB: $\OCSAT_\exists(\CloneE)$} &
      \cellcolor{lightgray}
      trivial,\newline
        $\TSAT_\exall(\CloneR_0)$ &
      \cellcolor{red}
      \EXPTIME-compl.,\newline
        LB: $\TSAT_\exists(\CloneN_2)$
      &
      \cellcolor{red}
      \EXPTIME-compl.,\newline 
        LB: \Cref{lem:TSAT_exists_M_hardness}
      &
      \cellcolor{lightgray}
      trivial,\newline
        $\OCSAT_\exall(\CloneR_1)$ &
      \cellcolor{lightgray}
      trivial,\newline 
        $\TSAT_\exall(\CloneR_0)$
      \\\hline
    $\TSAT_\exall$&
      \cellcolor{lightgray}
      trivial,\newline 
        $\TSAT_\exall(\CloneR_0)$ &
      \cellcolor{red}
      \EXPTIME-compl.,\newline 
        LB: \Cref{lem:ALCOCSAT_exall_I0_EXPTIME-HARD}&
      \cellcolor{red}
      \EXPTIME-compl.,\newline 
        LB: \Cref{lem:ALCOCSAT_exall_N2_EXPTIME-HARD}&
      \cellcolor{lightgray}
      trivial,\newline 
        $\TSAT_\exall(\CloneR_0)$ &
      \cellcolor{red}%
      \EXPTIME-compl.,\newline 
        LB: $\TSAT_\exall(\CloneI)$
      &
      \cellcolor{lightgray}
      trivial,\newline
        $\TSAT_\exall(\CloneR_0)$ &
      \cellcolor{red}
      \EXPTIME-compl,\newline 
        LB: $\TSAT_\exall(\CloneI)$&
      \cellcolor{lightgray}
      trivial,\newline
        $\TSAT_\exall(\CloneR_0)$ &
      \cellcolor{red}
      \EXPTIME-compl.,\newline
        LB: \Cref{lem:ALCOCSAT_exall_E_V_EXPTIME-hard} &
      \cellcolor{red}
      \EXPTIME-compl.,\newline 
        LB: $\TSAT_\exall(\CloneI)$
      &
      \cellcolor{lightgray}
      trivial,\newline
        $\OCSAT_\exall(\CloneR_1)$ &
      \cellcolor{lightgray}
      trivial,\newline 
        \Cref{lem:ALCTSAT_exall_R0_trivial}
      \\\hline\hline
    $\TCSAT_\emptyset$&
      \cellcolor{green}
      \NL-complete,\newline 
        LB: $\TSAT_\emptyset(\CloneI)$+LK,\newline
        UB: $\TCSAT_\emptyset(\CloneN)$ &
      \cellcolor{green}
      \NL-complete,\newline 
        LB: $\TSAT_\emptyset(\CloneI)$,\newline
        UB: $\TCSAT_\emptyset(\CloneN)$&
      \cellcolor{green}
      \NL-complete,\newline
        LB: $\TSAT_\emptyset(\CloneI)$,\newline
        UB: \Cref{TCSAT_emptyset_N_membership}
      &
      \cellcolor{blue}\white{%
      \P-complete,\newline 
        LB: $\TSAT_\emptyset(\CloneV)$+LK,\newline
        UB: $\OCSAT_\forall(\CloneV)$} 
      &
      \cellcolor{blue}\white{%
      \P-compl.,\newline 
        LB: $\TSAT_\emptyset(\CloneV)$,\newline
        UB: $\OCSAT_\forall(\CloneV)$}
      &
      \cellcolor{blue}\white{%
      \P-complete,\newline
        LB: $\TSAT_\emptyset(\CloneE)$+LK,\newline
        UB: $\OCSAT_\exists(\CloneE)$} &
      \cellcolor{blue}\white{%
      \P-complete,\newline 
        LB: $\TCSAT_\emptyset(\CloneE_0)$,\newline
        UB: $\OCSAT_\exists(\CloneE)$}&
      \cellcolor{orange}
      \NP-complete,\newline
        LB:$\TSAT_\emptyset(\CloneM)$+LK,\newline
        UB: \Cref{lem:OCSAT_emptyset_BF_membership} &
      \cellcolor{orange}
      \NP-complete,\newline 
        LB: $\TSAT_\emptyset(\CloneD)$,\newline 
        UB: \Cref{lem:OCSAT_emptyset_BF_membership}
      &
      \cellcolor{orange}
      \NP-complete,\newline
        LB: $\TSAT_\emptyset(\CloneM)$,\newline
        UB: \Cref{lem:OCSAT_emptyset_BF_membership}
      &
      \cellcolor{lightgray}
      trivial,\newline
        $\OCSAT_\exall(\CloneR_1)$ 
      &
      \cellcolor{orange}
      \NP-complete,\newline
        LB: $\TCSAT_\emptyset(\CloneS_{11})$,\newline
        UB: \Cref{lem:OCSAT_emptyset_BF_membership}
      \\\hline
    $\TCSAT_\forall$&
      \cellcolor{blue}\white{
      \P-complete,\newline 
        LB: $\TSAT_\forall(\CloneI)$+L.\ref{lem:Lewis_Trick2} \newline
        UB: $\TCSAT_\forall(\CloneV)$}
      &
      \cellcolor{blue}\white{
      \P-complete,\newline 
        LB: $\TSAT_\forall(\CloneI)$\newline
        UB: $\TCSAT_\forall(\CloneV)$}
        &
      \cellcolor{red}
      \EXPTIME-compl.,\newline 
        LB: $\TSAT_\forall(\CloneN_2)$
      &
      \cellcolor{blue}\white{%
      \P-complete,\newline 
        LB: $\TSAT_\emptyset(\CloneV)$+LK,\newline
        UB: $\OCSAT_\forall(\CloneV)$}
      &
      \cellcolor{blue}\white{%
      \P-compl.,\newline 
        LB: $\TSAT_\emptyset(\CloneV)$,\newline
        UB: $\OCSAT_\forall(\CloneV)$}
      &
      \cellcolor{red}
      \EXPTIME-compl.,\newline
        LB: \Cref{lem:TCSAT_forall_E0_hardness} &
      \cellcolor{red}
      \EXPTIME-compl.,\newline
        LB: $\TCSAT_\forall(\CloneE_0)$&
      \cellcolor{red}
      \EXPTIME-compl.,\newline
        LB: $\TCSAT_\forall(\CloneE_0)$ &
      \cellcolor{red}
      \EXPTIME-compl.,\newline 
        LB: $\TSAT_\forall(\CloneN_2)$
      &
      \cellcolor{red}
      \EXPTIME-compl.,\newline
        LB: $\TCSAT_\forall(\CloneE_0)$
      &
      \cellcolor{lightgray}
      trivial,\newline
        $\OCSAT_\exall(\CloneR_1)$ 
      &
      \cellcolor{red}
      \EXPTIME-compl.,\newline
        LB: $\TCSAT_\forall(\CloneS_{11})$
      \\\hline
    $\TCSAT_\exists$&
      \cellcolor{blue}\white{
      \P-complete,\newline 
        LB: $\TSAT_\exists(\CloneI)$+LK,\newline
        UB: $\TCSAT_\exists(\CloneE)$} &
      \cellcolor{blue}\white{
      \P-complete,\newline 
        LB: $\TSAT_\exists(\CloneI)$,\newline
        UB: $\TCSAT_\exists(\CloneE)$}
      &
      \cellcolor{red}
      \EXPTIME-compl.,\newline 
        LB: $\TSAT_\exists(\CloneN_2)$
      &
      \cellcolor{blue}\white{
      \P-complete,\newline 
        LB: $\TCSAT_\exists(\CloneI_0)$,\newline
        UB: $\TCSAT_\exists(\CloneV)$}
      &
      \cellcolor{blue}\white{
      \P-complete,\newline 
        LB: $\TCSAT_\exists(\CloneI)$,\newline
        UB: $\TSAT_\exists(\CloneV)$+L.\ref{lem:TCSAT_reduces_to_TSAT_with_true}}
      &
      \cellcolor{blue}\white{%
      \P-complete,\newline
        LB: $\TCSAT_\emptyset(\CloneE_0)$,\newline
        UB: $\OCSAT_\exists(\CloneE)$} &
      \cellcolor{blue}\white{%
      \P-complete,\newline
        LB: $\TCSAT_\emptyset(\CloneE_0)$,\newline
        UB: $\OCSAT_\exists(\CloneE)$}&
      \cellcolor{red}
      \EXPTIME-compl,\newline
        LB:$\TSAT_\exists(\CloneM)$+LK &
      \cellcolor{red}
      \EXPTIME-compl.,\newline 
        LB: $\TSAT_\exists(\CloneN_2)$
      &
      \cellcolor{red}
      \EXPTIME-compl.,\newline
        LB: $\TSAT_\exists(\CloneM)$
      &
      \cellcolor{lightgray}
      trivial,\newline
        $\OCSAT_\exall(\CloneR_1)$ 
      &
      \cellcolor{red}
      \EXPTIME-compl.,\newline
        LB: $\TCSAT_\exists(\CloneS_{11})$
      \\\hline
    $\TCSAT_\exall$&
      \cellcolor{red}
      \EXPTIME-compl.,\newline 
        LB: \Cref{lem:ALCOCSAT_exall_I0_EXPTIME-HARD} &
      \cellcolor{red}
      \EXPTIME-compl.,\newline 
        LB: $\TSAT_\exall(\CloneI)$ &
      \cellcolor{red}
      \EXPTIME-compl.,\newline 
        LB: $\TSAT_\exall(\CloneN_2)$ &
      \cellcolor{red}
      \EXPTIME-compl.,\newline 
        LB: \Cref{lem:ALCOCSAT_exall_E_V_EXPTIME-hard}&
      \cellcolor{red}
      \EXPTIME-compl.,\newline 
        LB: $\TCSAT_\exall(\CloneV_0)$
      &
      \cellcolor{red}
      \EXPTIME-compl.,\newline
        LB: \Cref{lem:ALCOCSAT_exall_E_V_EXPTIME-hard}&
      \cellcolor{red}
      \EXPTIME-compl.,\newline
        LB: $\TCSAT_\exall(\CloneI_0)$&
      \cellcolor{red}
      \EXPTIME-compl.,\newline
        LB: $\TCSAT_\exall(\CloneI_0)$&
      \cellcolor{red}
      \EXPTIME-compl.,\newline
        LB: $\TSAT_\exall(\CloneD)$&
      \cellcolor{red}
      \EXPTIME-compl.,\newline
        LB: $\TCSAT_\exall(\CloneI_0)$
      &
      \cellcolor{lightgray}
      trivial,\newline
        $\OCSAT_\exall(\CloneR_1)$ &
      \cellcolor{red}
      \EXPTIME-compl.,\newline
      LB: $\TCSAT_\exall(\CloneI_0)$
      \\\hline\hline
    $\OCSAT_\emptyset$&
      \cellcolor{green}
      \NL-complete,\newline 
        LB: $\TCSAT_\emptyset(\CloneI_0)$,\newline
        UB: $\OCSAT_\emptyset(\CloneN)$&
      \cellcolor{green}
      \NL-complete,\newline 
        LB: $\TCSAT_\emptyset(\CloneI_0)$,\newline
        UB: $\OCSAT_\emptyset(\CloneN)$&
      \cellcolor{green}
      \NL-complete,\newline
        LB: $\TCSAT_\emptyset(\CloneI_0)$,\newline
        UB: \Cref{lem:OCSAT_emptyset_N_membership}
      &
      \cellcolor{blue}\white{%
      \P-compl.,\newline 
        LB: $\TCSAT_\emptyset(\CloneV_0)$,\newline
        UB: $\OCSAT_\forall(\CloneV)$}&
      \cellcolor{blue}\white{%
      \P-compl.,\newline 
        LB: $\TCSAT_\emptyset(\CloneV_0)$,\newline
        UB: $\OCSAT_\forall(\CloneV)$}
      &
      \cellcolor{blue}\white{%
      \P-complete,\newline
        LB: $\TCSAT_\emptyset(\CloneE_0)$,\newline
        UB: $\OCSAT_\exists(\CloneE)$}&
      \cellcolor{blue}\white{%
      \P-complete,\newline
        LB: $\TCSAT_\emptyset(\CloneE_0)$,\newline
        UB: $\OCSAT_\exists(\CloneE)$}&
      \cellcolor{orange}
      \NP-complete,\newline
        LB: $\TCSAT_\emptyset(\CloneS_{11})$,\newline
        UB: \Cref{lem:OCSAT_emptyset_BF_membership}&
      \cellcolor{orange}
      \NP-complete,\newline
        LB: $\TSAT_\emptyset(\CloneD)$,\newline
        UB: \Cref{lem:OCSAT_emptyset_BF_membership}
      &
      \cellcolor{orange}
      \NP-complete,\newline
        LB: $\TCSAT_\emptyset(\CloneS_{11})$,\newline
        UB: \Cref{lem:OCSAT_emptyset_BF_membership}
      &
      \cellcolor{lightgray}
      trivial,\newline
        $\OCSAT_\exall(\CloneR_1)$ &
      \cellcolor{orange}
      \NP-complete,\newline
        LB: $\TCSAT_\emptyset(\CloneS_{11})$,\newline
        UB: \Cref{lem:OCSAT_emptyset_BF_membership}
      \\\hline
    $\OCSAT_\forall$&
      \cellcolor{blue}\white{
      \P-complete,\newline 
        LB: $\TCSAT_\forall(\CloneI_0)$
        UB: $\OCSAT_\forall(\CloneV)$}
        &
      \cellcolor{blue}\white{
      \P-complete,\newline 
        LB: $\TCSAT_\forall(\CloneI_0)$\newline
        UB: $\OCSAT_\forall(\CloneV)$}
      &
      \cellcolor{red}
      \EXPTIME-compl.,\newline 
        LB: $\TSAT_\forall(\CloneN_2)$
      &
      \cellcolor{blue}\white{%
      \P-complete,\newline 
        LB: $\TCSAT_\emptyset(\CloneV_0)$,\newline
        UB: $\OCSAT_\forall(\CloneV)$}&
      \cellcolor{blue}\white{%
      \P-complete,\newline 
        LB: $\TCSAT_\emptyset(\CloneV_0)$,\newline
        UB: \Cref{lem:OCSAT_forall_V_membership}}
      &
      \cellcolor{red}
      \EXPTIME-compl.,\newline
        LB: $\TCSAT_\forall(\CloneE_0)$&
      \cellcolor{red}
      \EXPTIME-compl.,\newline
        LB: $\TCSAT_\forall(\CloneE_0)$&
      \cellcolor{red}
      \EXPTIME-compl.,\newline
        LB: $\TCSAT_\forall(\CloneE_0)$&
      \cellcolor{red}
      \EXPTIME-compl.,\newline 
        LB: $\TSAT_\forall(\CloneN_2)$
      &
      \cellcolor{red}
      \EXPTIME-compl.,\newline
        LB: $\TCSAT_\forall(\CloneE_0)$
      &
      \cellcolor{lightgray}
      trivial,\newline
        $\OCSAT_\exall(\CloneR_1)$ 
      &
      \cellcolor{red}
      \EXPTIME-compl.,\newline
        LB: $\TCSAT_\forall(\CloneS_{11})$
      \\\hline
    $\OCSAT_\exists$&
      \cellcolor{blue}\white{
      \P-complete,\newline 
        LB: $\TCSAT_\exists(\CloneI_0)$,\newline
        UB: $\OCSAT_\exists(\CloneE)$}
      &
      \cellcolor{blue}\white{
      \P-complete,\newline 
        LB: $\TCSAT_\exists(\CloneI_0)$,\newline
        UB: $\OCSAT_\exists(\CloneE)$}
      &
      \cellcolor{red}
      \EXPTIME-compl.,\newline 
        LB: $\TSAT_\exists(\CloneN_2)$
      &
			&
      &
      \cellcolor{blue}\white{%
      \P-complete,\newline
        LB: $\TCSAT_\emptyset(\CloneE_0)$,\newline
        UB: $\OCSAT_\exists(\CloneE)$} &
      \cellcolor{blue}\white{%
      \P-complete,\newline
        LB: $\TCSAT_\emptyset(\CloneE_0)$,\newline
        UB: \Cref{lem:OCSAT_exists_E_membership}}&
      \cellcolor{red}
      \EXPTIME-compl.,\newline
        LB:$\TSAT_\exists(\CloneM)$+LK &
      \cellcolor{red}
      \EXPTIME-compl.,\newline 
        LB: $\TSAT_\exists(\CloneN_2)$
      &
      \cellcolor{red}
      \EXPTIME-compl.,\newline
        LB: $\TCSAT_\exists(\CloneM)$
      &
      \cellcolor{lightgray}
      trivial,\newline
        $\OCSAT_\exall(\CloneR_1)$ &
      \cellcolor{red}
      \EXPTIME-compl.,\newline
        LB: $\TCSAT_\exists(\CloneS_{11})$
      \\\hline
    $\OCSAT_\exall$&
      \cellcolor{red}
      \EXPTIME-compl,\newline 
        LB: $\TCSAT_\exall(\CloneI_0)$&
      \cellcolor{red}
      \EXPTIME-compl,\newline 
        LB: $\TCSAT_\exall(\CloneI_0)$&
      \cellcolor{red}
      \EXPTIME-compl,\newline 
        LB: $\TCSAT_\exall(\CloneN_2)$&
      \cellcolor{red}
      \EXPTIME-compl,\newline 
        LB: $\TCSAT_\exall(\CloneI_0)$&
      \cellcolor{red}
      \EXPTIME-compl,\newline 
        LB: $\TCSAT_\exall(\CloneI_0)$
      &
      \cellcolor{red}
      \EXPTIME-compl.,\newline
        LB: $\TCSAT_\exall(\CloneI_0)$&
      \cellcolor{red}
      \EXPTIME-compl.,\newline
        LB: $\TCSAT_\exall(\CloneI_0)$&
      \cellcolor{red}
      \EXPTIME-compl.,\newline
        LB: $\TCSAT_\exall(\CloneI_0)$&
      \cellcolor{red}
      \EXPTIME-compl.,\newline
        LB: $\TSAT_\exall(\CloneD)$&
      \cellcolor{red}
      \EXPTIME-compl.,\newline
        LB: $\TCSAT_\exall(\CloneI_0)$
      &
      \cellcolor{lightgray}
      trivial,\newline
        \Cref{lem:ALCOCSAT_exall_R1_trivial} 
      &
      \cellcolor{red}
      \EXPTIME-compl.,\newline
      LB: $\TCSAT_\exall(\CloneI_0)$
      \\
  \end{supertabular}
\end{center}}
\end{landscape}

  \begin{figure}[htpb]
  \centering
\ifasy
\begin{asy}
     import lattice;
     import patterns;
      import pens;
     //add("hatch",hatch(H=4, dir=NE, white+3.5));
     defaultpen(0.8+fontsize(9));
 
   //  pen ParityL_to_P = pattern("hatch");
 
     Lattice lattice = Lattice(0.9cm, 0.65cm, 0.28cm);
 
     lattice.setUpperBound(lattice.BF,black,EXPTIME,"$\mathsf{EXPTIME}$-complete",1);
 
     lattice.setLowerBound(lattice.I0,black,EXPTIME,1);
     lattice.setLowerBound(lattice.N2,black,EXPTIME,1);
 
     //Clone[] c = {lattice.I,lattice.BF,lattice.M,lattice.E,lattice.V,lattice.N,lattice.L};
     //lattice.restrictTo(c);
 
      lattice.setUpperBound(lattice.R0,black,EXPtriv,"$\mathsf{EXPTIME}$-compl. / trivial for $\mathrm{TSAT}$",2);
 
     lattice.setUpperBound(lattice.R1,black,trivial, "trivial",9);
     lattice.draw();
     lattice.legend(point(S)+(2.5cm,0.5cm));
\end{asy}
\else
\includegraphics{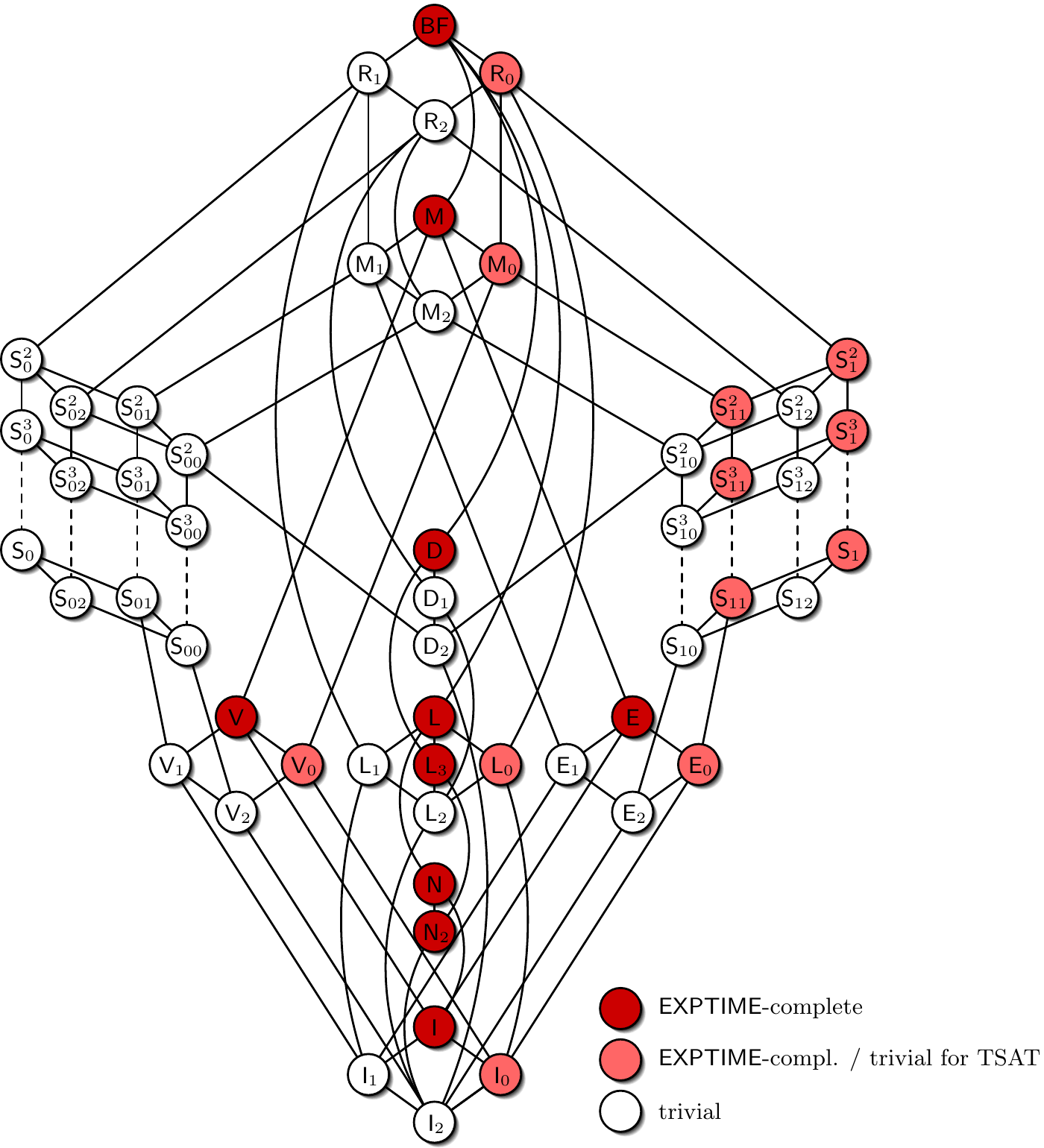}
\fi
  \caption{Complexity for $\ALCTSAT(B)$, $\ALCTCSAT(B)$, $\ALCOSAT(B)$ and $\ALCOCSAT(B)$.}
  \label{fig:latticeALCOCSAT}

  \end{figure}
  \begin{figure}[htpb]
  \centering
\ifasy  
\begin{asy}
     import lattice;
     import patterns;
     import pens;
     defaultpen(0.8+fontsize(9));
 
     Lattice lattice = Lattice(0.9cm, 0.65cm, 0.28cm);
 
     lattice.setUpperBound(lattice.BF,black,NP,"$\mathsf{NP}$-complete",1);
     lattice.setUpperBound(lattice.V,black,P, "$\mathsf{P}$-complete",2);
     lattice.setUpperBound(lattice.E,black,P, "$\mathsf{P}$-complete",2);
     lattice.setUpperBound(lattice.N,black,NL, "$\mathsf{NL}$-complete",3);
     lattice.setUpperBound(lattice.R1,black,trivial, "trivial",9);
     lattice.setUpperBound(lattice.R0,black,trivial,9);

     lattice.setLowerBound(lattice.L3,black,NP,1);
 
     lattice.draw();
     lattice.legend(point(S)+(3cm,0.5cm));
\end{asy}
\else
\includegraphics{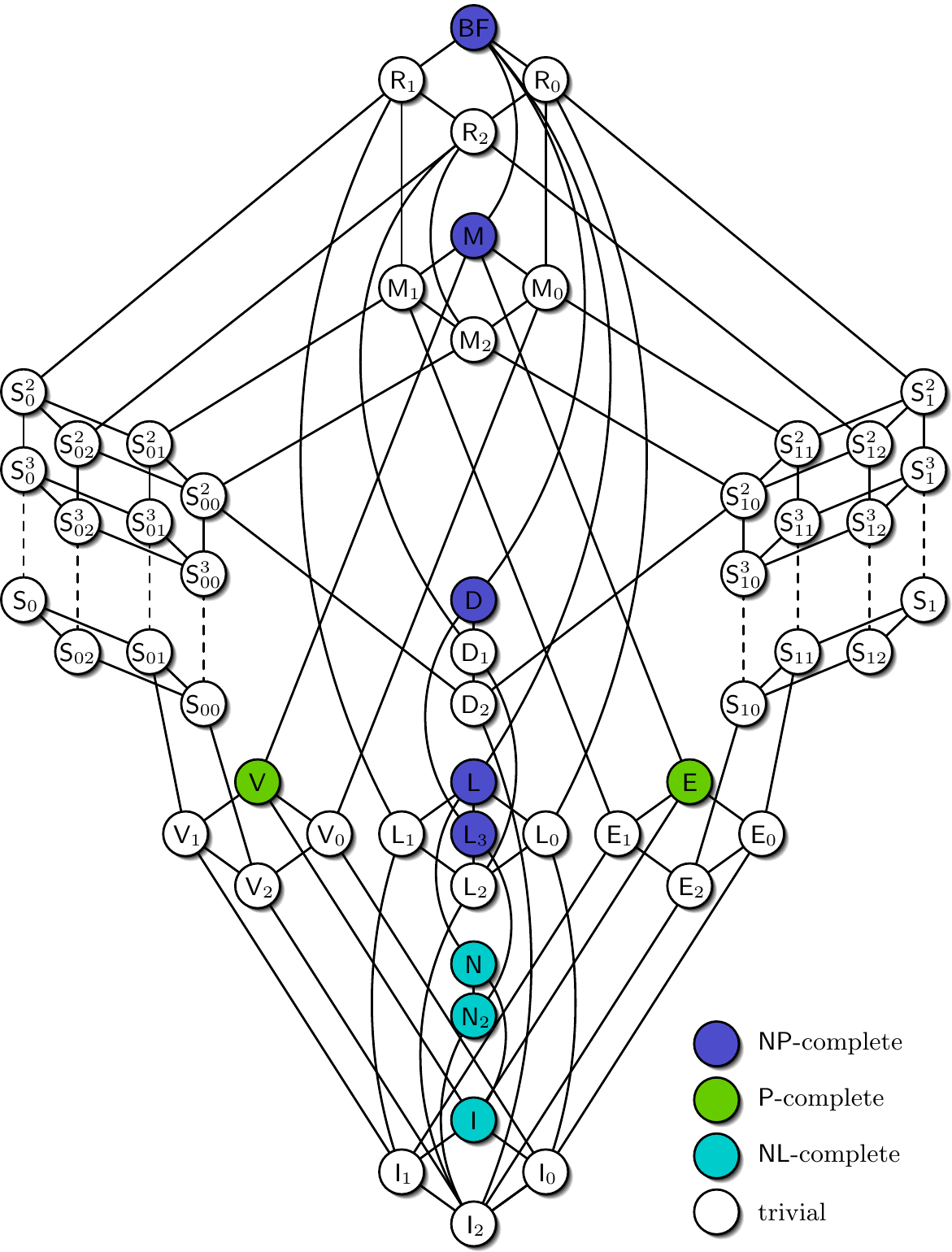}
\fi
  \caption{Complexity for $\ALCTSAT_\emptyset(B)$.}
  \label{fig:latticeTSATemptyset}
  \end{figure}

  \begin{figure}[htpb]
  \centering
\ifasy  
\begin{asy}
     import lattice;
     import patterns;
     import pens;
     defaultpen(0.8+fontsize(9));
 
     Lattice lattice = Lattice(0.9cm, 0.65cm, 0.28cm);

     lattice.setUpperBound(lattice.BF,black,EXPTIME,"$\mathsf{EXPTIME}$-complete",1);
     lattice.setUpperBound(lattice.V,black,P, "$\mathsf{P}$-complete",2);
     lattice.setUpperBound(lattice.E,black,P, "$\mathsf{P}$-complete",2);
     lattice.setUpperBound(lattice.I,black,P, 2);
     lattice.setUpperBound(lattice.R1,black,trivial, "trivial",9);
     lattice.setUpperBound(lattice.R0,black,trivial,9);

     lattice.setLowerBound(lattice.N2,black,EXPTIME,1); 
     lattice.setLowerBound(lattice.M,black,EXPTIME,1);
     
     lattice.draw();
     lattice.legend(point(S)+(3cm,0.5cm));
\end{asy}
\else
\includegraphics{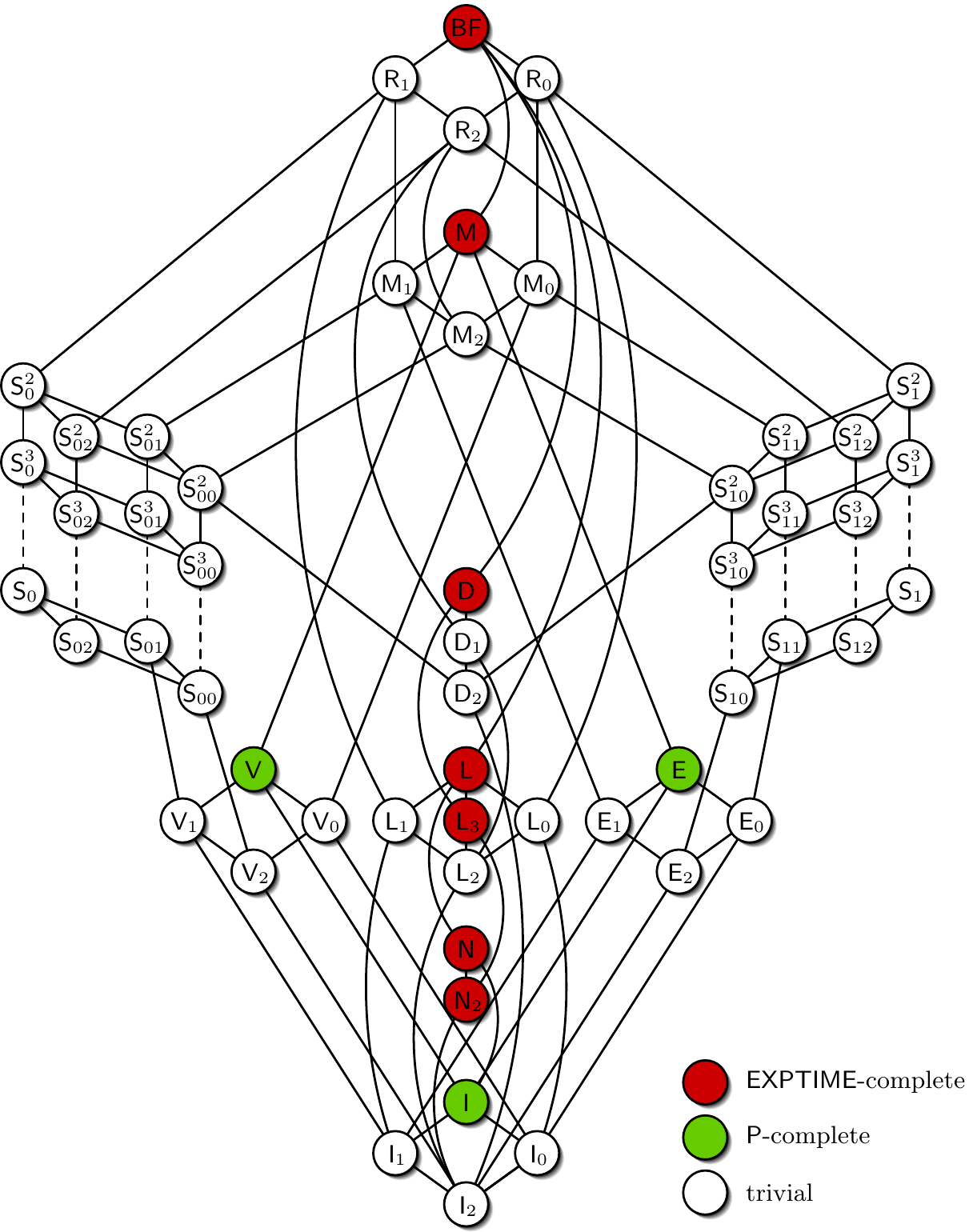}
\fi
  \caption{Complexity for $\ALCTSAT_\exists(B)$ and $\ALCTSAT_\forall(B)$.}
  \label{fig:latticeTSATexists}

  \end{figure}

  \begin{figure}[htpb]
  \centering
 
\ifasy
\begin{asy}
     import lattice;
     import patterns;
     import pens;
     defaultpen(0.8+fontsize(9));
 
     Lattice lattice = Lattice(0.9cm, 0.65cm, 0.28cm);
 
     lattice.setUpperBound(lattice.BF,black,NP,"$\mathsf{NP}$-complete",1);
     lattice.setUpperBound(lattice.V,black,P, "$\mathsf{P}$-complete",2);
     lattice.setUpperBound(lattice.E,black,P, "$\mathsf{P}$-complete",2);
     lattice.setUpperBound(lattice.N,black,NL, "$\mathsf{NL}$-complete",3);
     lattice.setUpperBound(lattice.R1,black,trivial, "trivial",9);

     lattice.setLowerBound(lattice.I0,black,NL,3); 
     lattice.setLowerBound(lattice.E0,black,P,2); 
     lattice.setLowerBound(lattice.V0,black,P,2);
     lattice.setLowerBound(lattice.L3,black,NP,1);
     lattice.setLowerBound(lattice.L0,black,NP,1);
     lattice.setLowerBound(lattice.S11,black,NP,1);     
     lattice.setLowerBound(lattice.D,black,NP,1);
     lattice.draw();
     lattice.legend(point(S)+(3cm,0.5cm));
\end{asy}
\else
\includegraphics{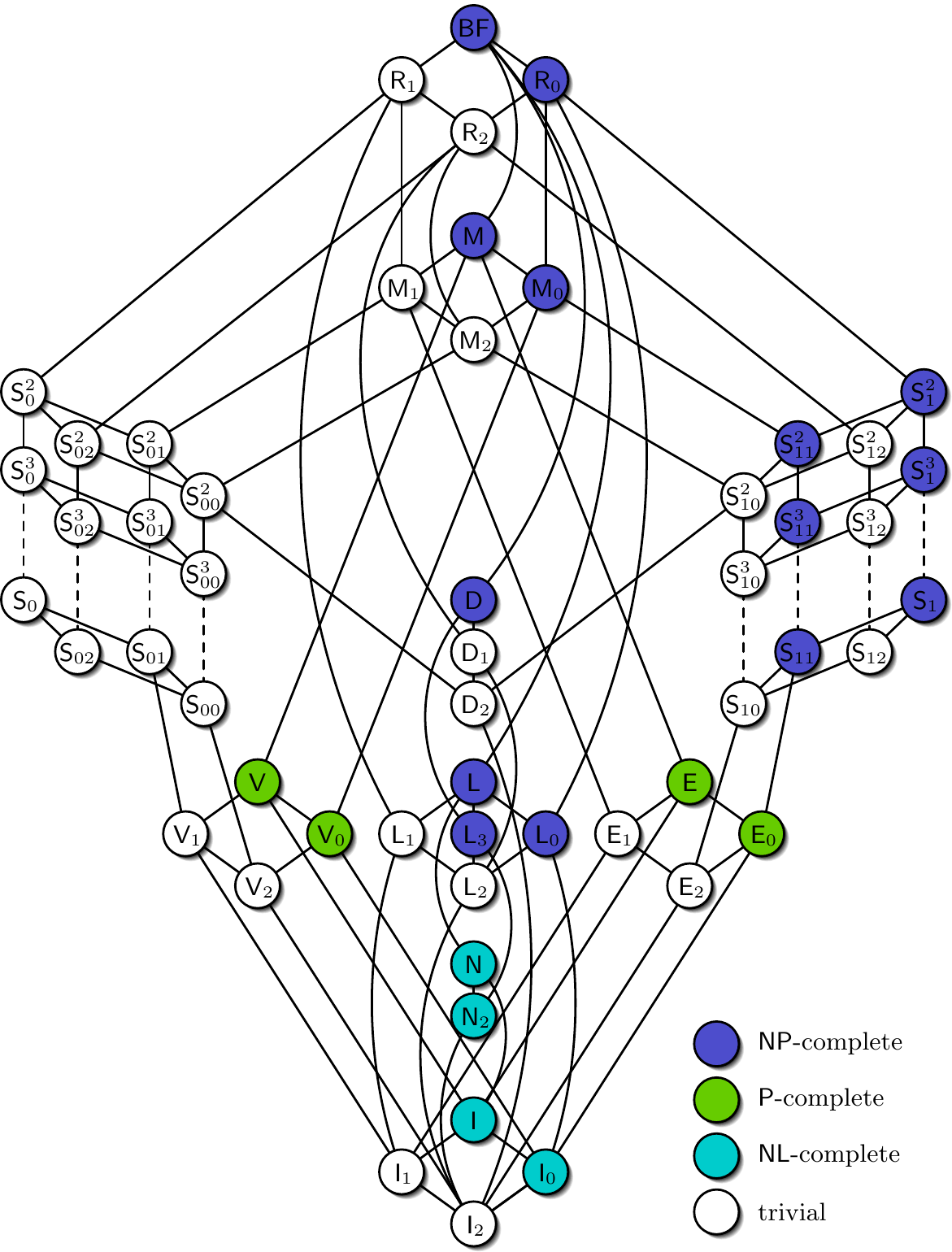}
\fi
  \caption{Complexity for $\cSTARSAT_\emptyset(B)$.}
  \label{fig:latticeSTARSATemptyset}

  \end{figure}

  \begin{figure}[htpb]
  \centering
  
\ifasy
\begin{asy}
     import lattice;
     import patterns;
     import pens;
     defaultpen(0.8+fontsize(9));
 
     Lattice lattice = Lattice(0.9cm, 0.65cm, 0.28cm);

     lattice.setLowerBound(lattice.I0,black,P,"$\mathsf{P}$-complete",2); 
     lattice.setLowerBound(lattice.E0,black,EXPTIME,"$\mathsf{EXPTIME}$-complete",1);
     lattice.setLowerBound(lattice.N2,black,EXPTIME,1);

     lattice.setUpperBound(lattice.BF,black,EXPTIME,1);
     lattice.setUpperBound(lattice.L0,black,EXPTIME,1);
     lattice.setUpperBound(lattice.V,black,P,2);
     lattice.setUpperBound(lattice.E,black,EXPTIME,1);
     lattice.setUpperBound(lattice.I,black,P,2);
     lattice.setUpperBound(lattice.R1,black,trivial, "trivial",9);

     lattice.draw();
     lattice.legend(point(S)+(3cm,0.5cm));
\end{asy}
\else
\includegraphics{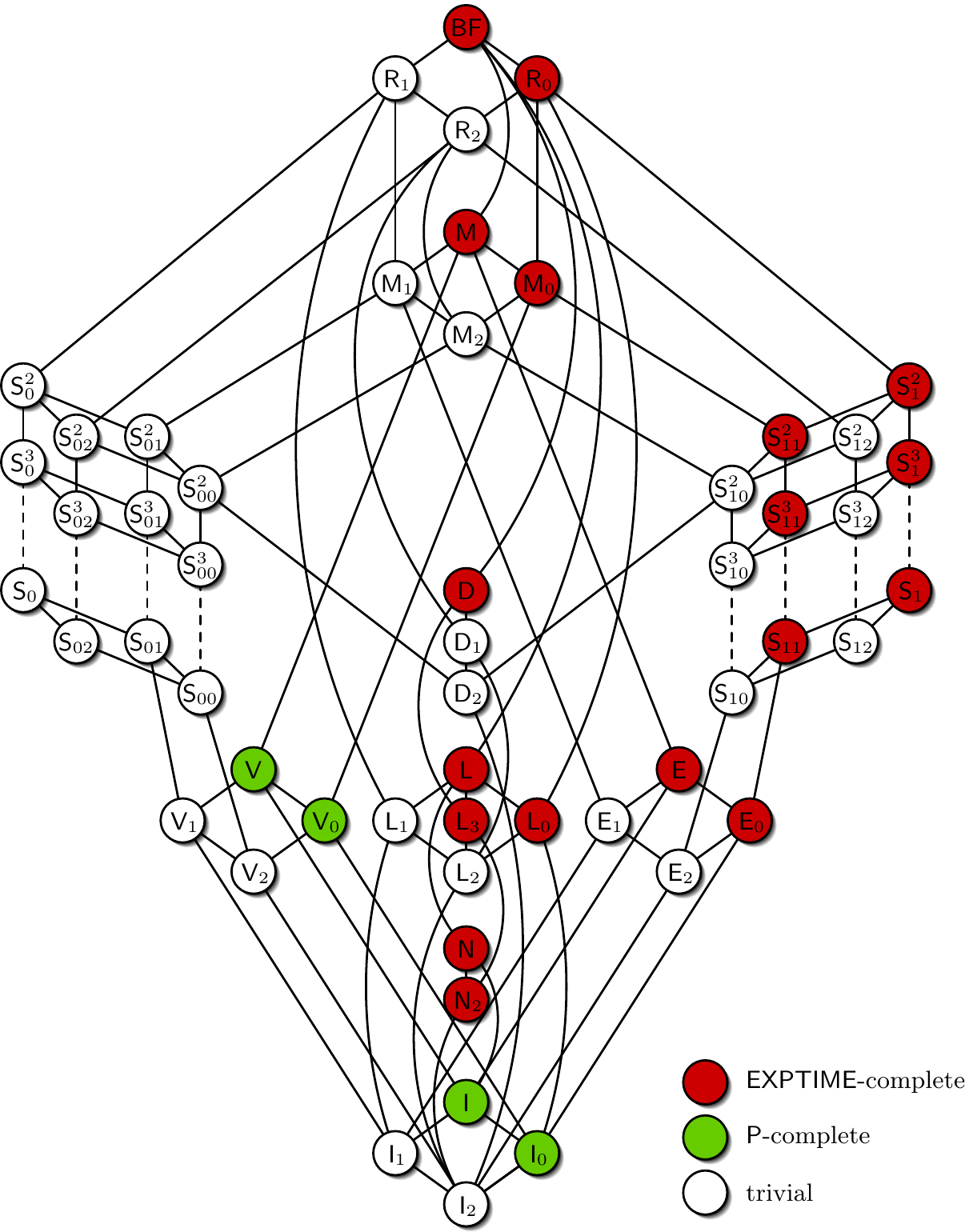}
\fi
  \caption{Complexity for $\cSTARSAT_\forall(B)$.}
  \label{fig:latticeSTARSATforall}

  \end{figure}

  \begin{figure}[htpb]
  \centering
  
\ifasy
\begin{asy}
     import lattice;
     import patterns;
     import pens;
     defaultpen(0.8+fontsize(9));
 
     Lattice lattice = Lattice(0.9cm, 0.65cm, 0.28cm);

     lattice.setUpperBound(lattice.BF,black,EXPTIME,"$\mathsf{EXPTIME}$-complete",1);
     lattice.setUpperBound(lattice.E,black,P,"$\mathsf{P}$-complete",2);
     lattice.setUpperBound(lattice.V,black,Punknown,
     "\hspace{-1.45ex}\begin{tabular}{l}for TCSAT$(B)$ $\mathsf{P}$-complete,\\else in $\mathsf{EXPTIME}$ and $\mathsf{P}$-hard\end{tabular}",10);
     lattice.setUpperBound(lattice.I,black,P,2);
     
     lattice.setLowerBound(lattice.I0,black,P,2); 
     lattice.setLowerBound(lattice.V0,black,Punknown,10); 
     lattice.setLowerBound(lattice.L0,black,EXPTIME,1); 
     lattice.setLowerBound(lattice.N2,black,EXPTIME,1); 
     lattice.setLowerBound(lattice.E0,black,P,2); 
     lattice.setLowerBound(lattice.S11,black,EXPTIME,1); 

     lattice.setUpperBound(lattice.R1,black,trivial,"trivial",9);

     lattice.draw();
     lattice.legend(point(S)+(3cm,0.5cm));
\end{asy}
\else
\includegraphics{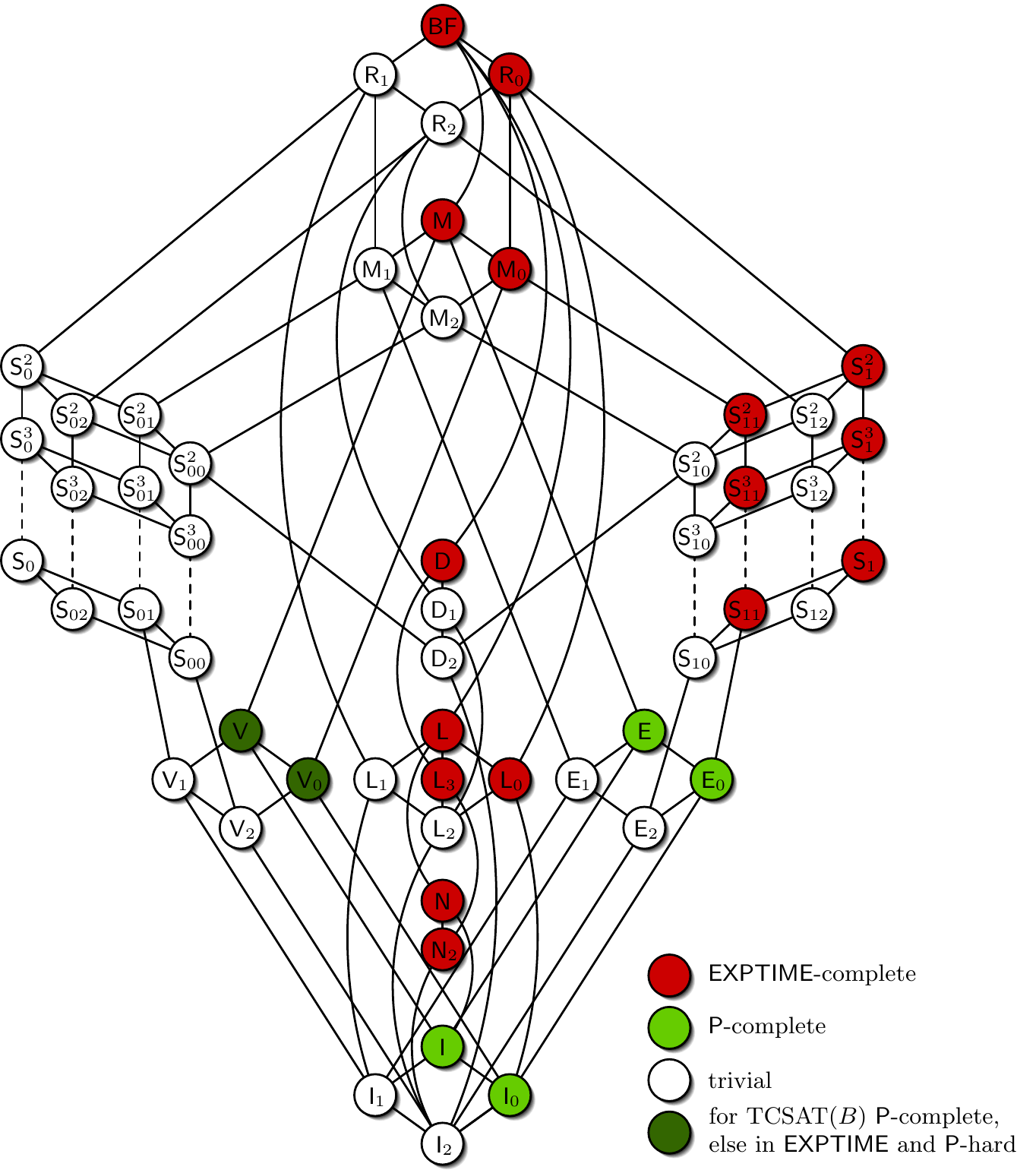}
\fi
  \caption{Complexity for $\cSTARSAT_\exists(B)$.}
  \label{fig:latticeSTARSATexists}

  \end{figure}
\newpage

\subsection*{Acknowledgements}
We thank Peter Lohmann for helpful comments and suggestions.

\bibliographystyle{plain}
\bibliography{description_logic}

\end{document}